\documentclass[conference,letter]{IEEEtran}
\IEEEoverridecommandlockouts
\usepackage{amsmath,amsfonts,amssymb}
\usepackage{algorithmic}
\usepackage{algorithm}
\usepackage{amsthm}
\usepackage{array}
\usepackage{mathtools}
\usepackage{textcomp}
\usepackage{stfloats}
\usepackage{url}
\usepackage{verbatim}
\usepackage{enumitem}
\usepackage{xcolor}
\usepackage{cleveref}
\usepackage{graphicx}
\usepackage{cite}
\newtheorem{definition}{\bf \emph{Definition}}
\newtheorem{theorem}{\bf \emph{Theorem}}
\newtheorem{lemma}{\bf \emph{Lemma}}

\newtheorem{proposition}{\bf \emph{Proposition}}
\theoremstyle{remark}
\newtheorem{remark}{\textit{Remark}}

 \def\cF{{\mathcal{F}}}  
\def\cI{{\mathcal{I}}}   \def\cL{{\mathcal{L}}}
 \def\cN{{\mathcal{N}}}  \def\cP{{\mathcal{P}}}
 \def\cR{{\mathcal{R}}} \def\cS{{\mathcal{S}}} 
   \def\cX{{\mathcal{X}}}
 
\def\das{\stackrel{a.s.}{=}}

\def\ddis{\stackrel{d}{=}}
\def\b({ \bigg( }
\def\b){ \bigg) }

\def\b[{\bigg[}
\def\b]{\bigg]}
\def\limn{\lim_{n \rightarrow \infty}}
\def\limN{\lim_{N \rightarrow \infty}}
\def\limsupn{\limsup_{n \rightarrow \infty}}
\def\limsupN{\limsup_{N \rightarrow \infty}}
\def\ba{{\mathbf{a}}} \def\bb{{\mathbf{b}}}    
    
  \def\bm{{\mathbf{m}}}  \def\bo{{\mathbf{o}}}
 \def\bq{{\mathbf{q}}}   
\def\bu{{\mathbf{u}}} \def\bv{{\mathbf{v}}} \def\bw{{\mathbf{w}}} \def\bx{{\mathbf{x}}} \def\by{{\mathbf{y}}}
 \def\bh{\mathbf{h}}

\def\bA{{\mathbf{A}}} \def\bB{{\mathbf{B}}}   
  \def\bH{{\mathbf{H}}} \def\bI{{\mathbf{I}}} 
  \def\bM{{\mathbf{M}}}  
\def\bP{{\mathbf{P}}} \def\bQ{{\mathbf{Q}}}   
 \def\bV{{\mathbf{V}}}  \def\bX{{\mathbf{X}}} \def\bY{{\mathbf{Y}}}

  \def\R{{\mathbb{R}}}        \def\E{\mathbb{E}}

\def\sumn{\sum_{i=1}^{n}}
\def\sumN{\sum_{i=1}^{N}}

\def\Cov{\mathop{\mathrm{Cov}}}

     \def\d4{\!\!\!\!}

\def\lL{\langle}  \def\lR{\rangle}

  \def\-{\! - \!}  \def\+{\! + \!}  \def\={\! = \!}  \def\>{\! > \!}

\newcommand{\bef}{\begin{figure}}
\newcommand{\eef}{\end{figure}}
\newcommand{\beq}{\begin{eqnarray}}
\newcommand{\eeq}{\end{eqnarray}}

\begin{document}
\title{On the Stability of Approximate Message Passing with Independent Measurement Ensembles 
}
\author{\IEEEauthorblockN{Dang Qua Nguyen and Taejoon Kim}
	
	\IEEEauthorblockA{\textit{Department of Electrical Engineering and Computer Science} \\
		\textit{The University of Kansas, Lawrence, KS 66045 USA}\\
		Email: \{quand, taejoonkim\}@ku.edu \vspace{-0.5cm} }
  \thanks{This work was supported in part by the National Science Foundation (NSF) under Grant CNS1955561, CNS2212565, CNS2225577, and the Office of Naval Research (ONR) under Grant N00014-21-1-2472.}
}
\maketitle	
\begin{abstract}
Approximate message passing (AMP) is a scalable, iterative approach to signal recovery. 
For structured random measurement ensembles, including independent and identically distributed (i.i.d.) Gaussian and rotationally-invariant matrices, the performance of AMP can be characterized by a scalar recursion called state evolution (SE). 
The pseudo-Lipschitz (polynomial) smoothness is conventionally assumed. 
In this work, we extend the SE for AMP to a new class of measurement matrices with independent (not necessarily identically distributed) entries.
We also extend it to a general class of functions, called controlled functions which are not constrained by the polynomial smoothness; unlike the pseudo-Lipschitz function that has polynomial smoothness, the controlled function grows exponentially. 
The lack of structure in the assumed measurement ensembles is addressed by leveraging Lindeberg-Feller. 
The lack of smoothness of the assumed controlled function is addressed by a proposed conditioning technique leveraging the empirical statistics of the AMP instances. 
The resultants grant the use of the SE to a broader class of measurement ensembles and a new class of functions.
\end{abstract}
\section{Introduction}
\label{secI}	
    The problem of signal recovery from a linear observation\footnote{ A bold lower case letter $\ba$ is a column vector and a bold upper case letter $\bA$ is a matrix. 
    $\|\ba\|_p$, $\bA^*$, {$\bA^{-1}$}, and $A_{ij}$ denote the $p$-norm of $\ba$, transpose of $\bA$, {inverse of $\bA$}, and $i$th row and $j$th column entry of $\bA$, respectively. 
    $\bA(N)$ and $\ba(N)$, respectively, are the matrix and vector indexed by $N$.   
    $\cN(\nu,\sigma^2)$ denotes the Gaussian distribution with mean $\nu$ and variance $\sigma^2$.
    The $\xRightarrow{d}$ and $\das$ denote the convergence in distribution and equivalence in an almost sure sense, respectively.
    $\lL\bu,\bv \lR = \frac{1}{n}\sumn u_iv_i$ defines the normalized inner product of  $\bu, \bv \in \R^{n \times 1}$.
    $\E_{Z}[\cdot]$ denotes the expectation with respect to the random variable $Z$.
    $\mathbf{0}_n$ denotes the $n \times 1$ all-zero vector.}   
        \vspace{-0.2cm}
        \begin{equation}
         \label{eq:5}
            \by = \bA\bx_0 +\bw
        \vspace{-0.1cm}
        \end{equation} 
        appears in various fields \cite{Kay97,Donoho2006,Schniter2022}, where $\by  \in \R^{n \times 1}$, $\bA \in \R^{n \times N}$ is a given measurement matrix, $\bx_0 \in \R^{N \times 1}$ is the signal to be recovered, and $\bw \in \R^{n \times 1}$ is an additive noise. 
        Of particular interest is the case when $\bA$ is overcomplete ($n \ll N$).    
        However, the computational cost to solve this problem is typically prohibitive when the dimensions $n$ and $N$ are large. 
        Message passing (MP) can be applied to handle the large-dimensionality of the problem. 
        Conventionally, sparsity is essential for MP to approach a fundamental performance limit \cite{Richarson2001}.
        Recently, the approximate message-passing (AMP) algorithm has received significant attention \cite{Donoho2009,Donoho2010,Donoho20102} because it performs surprisingly well in systems that are not sparse.
    The remarkable features of AMP have inspired a wide range of applications \cite{Montanari2017, Fletcher2018, Lesieur2017, Pandit2019, Pandit2020, Emami2020, Zhang8, Jeon2, Schniter2010, Kim2015, Duan2015, Bellili2018,baron2020, Sung3, Tan2015, Millard2020}. 
    Despite its widespread applicability, the AMP algorithm suffers from instability issues \cite{Caltagirone2014,Ma2017,Sundeep19}.
    The instability is closely related to the underlying structure of the random measurement matrix $\bA$. 
    Understanding the dynamics of AMP for various classes of random measurement  ensembles has been an outstanding open problem.	
    
    A rigorous proof of state evolution (SE) was first established by Bayati and Montanari \cite{Bayati11}. 
    As $N$ tends to infinity while $\rho = \frac{n}{N}$ is kept constant, Bayati and Montanari \cite{Bayati11} asymptotically characterized the evolution of the mean squared error (state) of AMP for the  $\bA$  with independent and identically distributed (i.i.d.) Gaussian entries. 
    Rush \emph{et al.} \cite{Rush2018} showed a concentration bound of the SE in the finite $n$ and $N$ regime; the probability of deviation decays exponentially with $n$. 
    The validity of SE in \cite{Bayati11} has inspired extensive research efforts on extending it to different measurement ensembles such as sub-Gaussian $\bA$ by Bayati \emph{et al.} \cite{Mohsen15} and Chen \emph{et al.} \cite{chen2020}, right-orthogonally-invariant $\bA$ by Rangan \emph{et al.} \cite{Rangan2019}, rotationally-invariant $\bA$ by Fan \cite{Zhou2020}, unitarily-invariant $\bA$ by Takeuchi \cite{Takeuchi2020}, {and semi-random $\bA$ by Dudeja \emph{et al.} \cite{Dudeja2022}}.
    A belief is that the SE for AMP might hold for an even boarder class of matrices.  
   The key to analyzing the SE for AMP is the conditioning technique \cite{Bolthausen2009, Bayati11}.
   This means that the current instance of the AMP algorithm is modeled as a linear combination of the previous instances that are Gaussian, plus a deviation term (non-Gaussian). 
   A key step to the stability is leveraging the polynomial smoothness of the pseudo-Lipschitz function\footnote{The polynomial smoothness of the pseudo-Lipschitz function is defined in Appendix~\ref{PLdef}} and establishing that the contribution from the deviation term decays as $n$ and $N$ grow.   
    Recently, the controlled function has been introduced to analyze the nonlinear behavior of neural networks in machine learning \cite{yangTP1,yang2019}.
    Unlike the pseudo-Lipschitz function, the controlled function incorporates exponential growth into its model. 
    Hence, the pseudo-Lipschitz function can be viewed as a special case of the controlled function.
    This work is motivated by {the experimental observations \cite{Donoho2009,Ma2017}} that there is ample room for extending the establishment of SE for AMP to different classes of measurement ensembles and functions.  
    The major contributions of this work are summarized below.
    \vspace{-0.0cm}
    \begin{figure}[htp]\centering \includegraphics[width=0.41\textwidth]{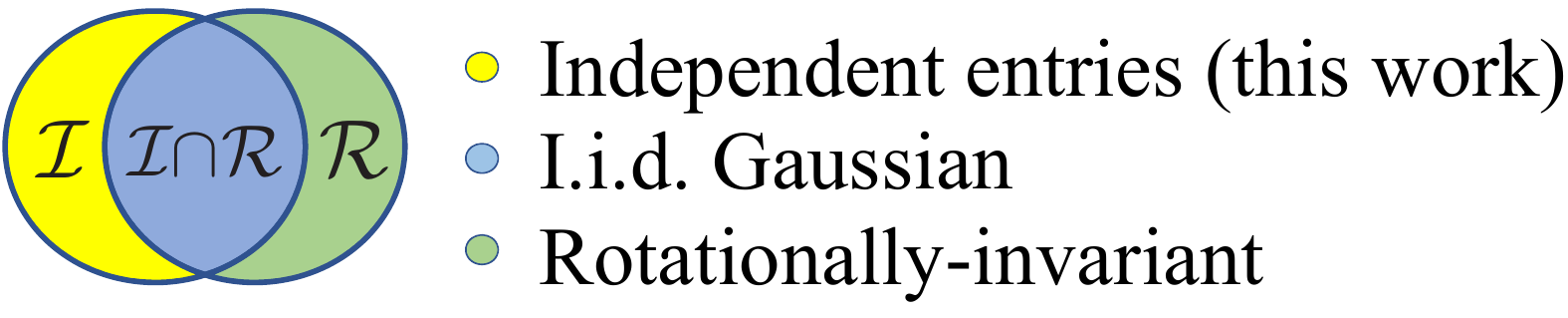} \vspace* {-0.0cm}
     \caption{Two categories of random measurement ensembles}
    \label{Fig1}
    \end{figure} 
    \vspace{-0.3cm}
    \begin{itemize}[leftmargin = 3mm]
    \item 
    We extend the SE analysis for AMP to the measurement ensembles in $\cI$ in Fig.~\ref{Fig1}, where $\cI$ is the set of matrices with independent but  not necessarily identically distributed entries. 
    Note that the SE analysis in the prior works \cite{Bayati11,Rush2018,Takeuchi2020,Rangan2019,Zhou2020} focused on the set $\cR$ in Fig.~\ref{Fig1}, where $\cR$ is the set of rotationally-invariant matrices, including the i.i.d. Gaussian ensembles. 
    The set $\cI$ distinguishes it from the set $\cR$ because it is far more incoherent than $\cR$.
    This constitutes a major challenge in establishing the SE 
    of the assumed ensembles.
    In our work, the lack of structure of the $\bA \in \cI$ is addressed by leveraging and extending the Lindeberg-Feller theorem.
    \item 
    In contrast to the prior works \cite{Bayati11,Rush2018, Mohsen15,chen2020,Takeuchi2020,Rangan2019,Zhou2020} 
    that relied on the pseudo-Lipschitz smoothness to establish SE, we generalize the SE analysis to the controlled function.
    The lack of smoothness of the latter is addressed by exploiting the measurability against Gaussian measure in conjunction with the conditioning technique, based on empirical statistics.  
    \end{itemize}
    A part of our results is a theoretical justification for the conjecture made in \cite{Bayati11} about the validity of the SE for AMP with i.i.d. non-Gaussian measurement ensembles. 
\section{Premilinaries} 
\label{secII}
First, the concept of empirical statistics is developed. Next, the frequently used statistical lemmas are presented.  
\vspace{-0.3cm}
\subsection{Empirical Statistics}\label{secIIA}
    The empirical law of a random vector constructed here is exploited to propose a conditioning technique in the following sections. We start by defining probability distributions.
    Suppose $\cP(\R)$ is the collection of all probability distributions on $\R$ with sample space $\Omega$. 
    A random variable \small$X:\Omega \rightarrow \R$\normalsize~has the distribution $\mu \in \cP(\R)$ denoted by \small$X \sim \mu$\normalsize, if \small$Pr(X \in \cS) = \mu(\cS)$\normalsize, for a set $\cS \subseteq \Omega$.
    For a function $f:\R \rightarrow \R$, we denote, if it exists, $\mu f = \int_{\R} f(x)\mu(dx)$.
    The first moment and second moment of the distribution $\mu$ are then given by taking \small$f(X) = X$\normalsize~and \small$f(X) = X^2$\normalsize, respectively, denoted as $\E[X] = \lL\mu \lR$ and $\E[X^2] = \lL\mu^2\lR$.
    The variance of the distribution $\mu$ is denoted by $\lL\mu\lR_2$, which equals to the variance of $X$. 
    Given these definitions, we define empirical distribution of deterministic vectors as follows. 

    \subsubsection{Empirical Distribution of a Deterministic Vector}
    For a deterministic vector $\bv \in \R^{n \times 1}$, the empirical sample mean, sample second moment, and sample variance are defined as \small$\lL\bv\lR \!=\!\frac{1}{n}\!\sum_{i=1}^{n}v_i$,\! $\lL\bv^2\lR \!=\! \frac{1}{n}\!\sum_{i=1}^{n}v^2_i$\normalsize, and \small$\lL\bv\lR_2 \!=\! \frac{1}{n}\sum_{i=1}^{n}(v_i-\lL\bv\lR)^2$\normalsize, respectively.
    We let $\delta_{\{a\}} \in \cP(\R)$ be the Dirac distribution with mass on $\{a\}$ {and \small$\delta_{\{a\}}f \triangleq f(a)$\normalsize~for all function $f:\R \rightarrow \R$}. 
    The empirical distribution of $\bv$ is then defined by $\widehat{\bv} \!\!=\!\! \frac{1}{n}\sum_{i=1}^{n}\delta_{\{v_i\}}$. 
    For the empirical distribution $\widehat{\bv}$, the first moment $\lL \widehat{\bv} \lR  \!\!=\!\!  \widehat{\bv}x  \!\!=\!\! \frac{1}{n} \sumn \delta_{\{v_i\}}x \!\!=\!\! \frac{1}{n}\sumn v_i =  \lL \bv \lR$, the second moment $\lL \widehat{\bv}^2 \lR  = \widehat{\bv}x^2 = \frac{1}{n}\sumn v^2_i= \lL \bv^2 \lR$, and the variance $\lL\widehat{\bv}\lR_2 = \frac{1}{n} \sumn (v_i-\lL\widehat{\bv}\lR)^2 = \lL\bv\lR_2$. 
    We are now ready to present the empirical law of a random vector. 
    \subsubsection{Empirical Law of a Random Vector}
    For a random vector $\bv:\Omega \rightarrow \R^{n \times 1}$, $\bv(\omega) \in \R^{n \times 1}$ is a random sample of the empirical distribution $\widehat{\bv(\omega)} \in \cP(\R)$. 
    The empirical law of the random vector $\bv$ is then defined by $\widehat{\bv} f = \E[\widehat{\bv(\omega)}f]$, for a measurable function $f$.
    Suppose that $v_i \sim \mu_i \in \cP(\R)$, $\forall i$. Then, we have the empirical law $\widehat{\bv} = \frac{1}{n}\sum_{i=1}^{n}\mu_i$.
    Moreover, the first moment $\lL \widehat{\bv}\lR = \frac{1}{n} \sum_{i=1}^{n}\lL \mu_i\lR$ and the second moment $\lL \widehat{\bv}^2\lR = \frac{1}{n} \sum_{i=1}^{n}\lL \mu_i^2\lR$. 
    \subsection{Frequently Used Statistics Results}\label{secIIB}    
    \begin{lemma}(Lindeberg-Feller \cite{Durrett10})\label{lm1}
    For $n \geq 1$, we let \small$\{X_{n,m}: 1\leq m\leq n\}$\normalsize~be an independent triangular array with $\E[X_{n,m}]\!=\!0$, $\forall m$. Suppose i) $\lim_{n\rightarrow \infty} \sum_{m=1}^{n}\E[X_{n,m}^2] = \sigma^2_x  > 0$; ii) (Lindeberg condition) for all $\epsilon > 0$, \small $\lim_{n \rightarrow \infty}\sum_{m=1}^{n}\mathbb{E}\left[X^2_{n,m};|X_{n,m}| > \epsilon\right] = 0$\normalsize. 
        \noindent Then $\sum_{m=1}^{n}X_{n,m} \xRightarrow{d} \cN(0,\sigma^2_x)$ as $n \rightarrow \infty$.
    \end{lemma}
     \noindent We propose a sufficient condition for the \emph{Lindeberg condition}.
    \begin{proposition}\label{ps1} 
    If there exists $\alpha >0$ such that \small$\sup_{m}\! \E[X^{2+2\alpha}_{n,m}] = o(n^{-1})$\normalsize, the Lindeberg condition in Lemma~\ref{lm1} holds. 
    \end{proposition}
    \begin{proof}
    See Appendix~\ref{ProofPS1}.
    \end{proof} 
 \noindent Incorporating {Proposition}~\ref{ps1} into {Lemma}~\ref{ps1} leads to a variant of Lindeberg-Feller that we use for our analysis.  
 \begin{proposition}\label{ps2}
    For $n \geq 1$, we let \small$\{X_{n,m}\!:\!1 \leq m \leq n\}$\normalsize~be independent random triangular array with \small$\E[X_{n,m}] \!=\! 0$, $\forall m$\normalsize. 
    Suppose   
    i) $\lim_{n \rightarrow \infty}$ $\sum_{m=1}^{n} \E[X^2_{n,m}] = \sigma^2_x  < \infty$; 
    ii) $\sup_m \mathbb{E}\left[X^{2+2\alpha}_{n,m}\right] = o(n^{-1})$, for a constant  $\alpha >0$.
    Then $\sum_{m=1}^{n}X_{n,m} \xRightarrow{d} \mathcal{N}(0,\sigma^2_x)$ as $n \rightarrow \infty$.
\end{proposition}
 \noindent {Proposition}~\ref{ps2} will be further exploited to propose our key propositions in Section~\ref{MCR}.
In what follows, $\bI$ denotes an identity matrix with appropriate dimensions. For a matrix $\bH \in \R^{n \times t}$ ($n \geq t$), \small$\bP_{\bH} = \bH(\bH^*\bH)^{-1}\bH^*$\normalsize ~is the orthogonal projection onto the column subspace of $\bH$ {and $\bP^{\perp}_{\bH} = \bI - \bP_{\bH}$}.
Two random variables $X$ and $Y$ are said to be equal in distribution, denoted by $X \ddis Y$, if \small$\E[\phi(X)Z] =\E[\phi(Y)Z]$\normalsize, 
for any integrable function {$\phi$} and random variable $Z$. 
The following proposition characterizes our  conditional distribution result. 
  \begin{proposition} \label{ps4} 
   Given $\bA \in \R^{n \times N}$, $\widetilde{\bA} \in \R^{n\times N} $ such that $\bA \stackrel{d}{=} \widetilde{\bA}$, and the set \small$\cP_{\bX,\bY} = \{\bA|\bA^*\bM = \bX, \bA\bQ =\bY\}$\normalsize, where $\bM \in \R^{n \times t}$, $\bX \in \R^{N \times t}$, $\bQ \in \R^{N \times t}$, and $\bY \in \R^{n \times t}$, the following holds 
 \vspace{-0.3cm}
 \small  
\begin{equation}
    \label{eq:Adistribution}
    \bA|_{\cP_{\bX,\bY}} \stackrel{d}{=} \bP^{\perp}_{\bM}\widetilde{\bA}\bP^{\perp}_{\bQ} + \bB,
\vspace{-0.15cm}
\end{equation} 
\normalsize
where $\bA|_{\cP_{\bX,\bY}}$ is the orthogonal projection of $\bA$ onto $\cP_{\bX,\bY}$ and \small$\bB= \bY(\bQ^*\bQ)^{-1}\bQ^*$ + $\bM(\bM^*\bM)^{-1}\bX^* -  \bM(\bM^*\bM)^{-1}\bX^*\bQ(\bQ^*\bQ)^{-1}\bQ^*$\normalsize.
\end{proposition}
\begin{proof}
See Appendix~\ref{ProofPS3}.    
\end{proof}
\begin{remark} An equivalence to \eqref{eq:Adistribution} was proven in \cite[Lemma 10, Eqn. (3.29)]{Bayati11}. 
\vspace{-0cm}
The underlying assumption of \cite{Bayati11} was that $A_{ij}$ are i.i.d. following $\cN(0,\frac{1}{n})$, in which the conditional distribution of $\bA$ is computed based on its  unitary-invariance property (i.e., $\bA \in \cI \cap \cR$). 
On the other hand, {Proposition}~\ref{ps4} does not rely on the unitary-invariance property.
The result in Proposition~\ref{ps4} will be exploited in Section~\ref{secIII}--\ref{secIV} to characterize the conditional distribution of $\bA \in \cI$ in our SE analysis.  
\end{remark}

\section{Main Results}
\label{secIII}
In this section, we present our main results by showing the SE for AMP when $\bA \in \cI$ and for the controlled function. 

Given the linear observation in \eqref{eq:5}, the general AMP algorithm \cite{Bayati11} recurs the vectors $\bh^{t+1} \in \R^{N \times 1}$, $\bq^t \in \R^{N \times 1}$, $\bb^t \in \R^{n \times 1}$, and $\bm^t \in \R^{n \times 1}$, sequentially, for $t \geq 0$, through    
    \vspace{-0.1cm}
    \small\begin{subequations}
    \label{eq:6}
    \beq
    \bq^t = f_t(\bh^t, \bx_{0}),&& \bb^t =\bA\bq^t - \lambda_t\bm^{t-1} \label{eq:6a},\\
    \bm^t =g_t(\bb^t,\bw),&& \bh^{t+1} = \bA^*\bm^t - \xi_t\bq^t \label{eq:6b},
    \eeq 
    \end{subequations}\normalsize
    where \small$\bq^0 \in \R^{N \times 1}$\normalsize~is an initial condition,  \small$\bh^0 = \mathbf{0}_N$, $\bm^{-1} = \mathbf{0}_n$\normalsize, and $f_t(\cdot,\cdot)$ and $g_t(\cdot,\cdot)$ are {controlled functions}.
    The $f_t$ and $g_t$ are applied entry-wise when their arguments are vectors, e.g., \small$g_t(\bb^t,\bw) = (g_t(b^t_1,w_1),\dots,g_t(b^t_n,w_n)) \in \R^{n \times 1}$\normalsize. 
    Suppose that  $f_t$ and $g_t$ are differentiable with respect to their first argument. 
    Then, the scalar $\lambda_t$ and $\xi_t$ in \eqref{eq:6} are, respectively, \small$\lambda_t =\frac{1}{\rho} \lL f'_t(\bh^t,\bx_0)\lR$\normalsize~and \small$\xi_t = \lL g'_t(\bb^t,\bw)\lR$\normalsize, where \small$f'_t(\bh^t,\bx_0) = \left(\frac{\partial f_t}{\partial h^t_1},\dots,\frac{\partial f_t}{\partial h^t_N}\right) \in \R^{N \times 1}$, $g'_t(\bb^t,\bw) = \left(\frac{\partial g_t}{\partial b^t_1},\dots,\frac{\partial g_t}{\partial b^t_n}\right)\in \R^{n \times 1}$,  $\frac{\partial f}{\partial x}$\normalsize~denotes the partial derivative of $f$ with respect to $x$, and $\rho = \frac{n}{N}$ is kept being constant as $n$ and $N$ tend to infinity. 
 \subsection{Definitions and Assumptions}\label{Sec:Assumption}
 \subsubsection{Controlled Function} A function $\phi$: $\R^{t \times 1} \rightarrow \R$ is called a controlled function \cite{yangTP1,yang2019} if for $\bx \in \R^{t \times 1}$,
        \vspace{-0.0cm}
    \small\begin{equation}
        \label{eq:controlled1}
        |\phi(\bx)| \leq c_1 \exp(c_2\|\bx\|^{\lambda}_2), 
    \end{equation}\normalsize
    where $c_1 > 0,  c_2 > 0,$ and $ 1 \leq \lambda <2$ are constant. 
    The controlled function is a general model of a nonlinear function without  polynomial smoothness constraint. 
    Although the controlled function may increase exponentially, it is integrable in $\cL^1$ and $\cL^2$ spaces against Gaussian measure \cite{yangTP1,yang2019}.
    Herein, the space $\cL^p(\cX,\cF,\mu)$ consists of all measurable functions $\{f\}$ on $\cX$ such that $\int_{\cX} |f(x)|^p d\mu(x) < \infty,$ where $(\cX,\cF,\mu)$ is a measure space and $p \in \{1,2\}$. 
     Using the inequalities $\|\bx\|_2 \leq \|\bx\|_{\lambda} \leq t^{\frac{1}{\lambda}-\frac{1}{2}}\|\bx\|_2$ for $\bx \in \R^{t\times1}$ and $1 \leq \lambda <2$, we get from \eqref{eq:controlled1}
     \vspace{-0.1cm}
     \small\begin{equation}
         \label{eq:controlled2}
         |\phi(\bx)|\leq c_1\exp(c_2\|\bx\|^\lambda_{\lambda}) \leq c_1\exp(c_2t^{\frac{1}{\lambda}-\frac{1}{2}}\|\bx\|^{\lambda}_2),
     \end{equation}\normalsize 
    revealing that \small$c_1\exp(c_2\|\bx\|^\lambda_{\lambda})$\normalsize~is also a controlled function.
    \subsubsection{Measurement Matrix}
    \label{definitionA}
    The matrix $\bA$ in \eqref{eq:5} consists of independent entries $A_{ij} \sim \frac{1}{\sqrt{n}}\mu_{ij}$, with $\lL \mu_{ij}\lR = 0$ and $\lL \mu_{ij} \lR_2 = 1$, which are not necessarily identically distributed, $\forall i,j$, i.e., $\bA \in \cI$.  
      
\subsubsection{Signal $\bx_0$ and Noise $\bw$}

The entries of signal vector {$\bx_0$} and noise $\bw$ in \eqref{eq:5} are i.i.d. according to the distribution $\mu_{X_0}$ and $\mu_{W}$,  respectively, where $\lL \mu_{X_0} \lR = 0, \lL \mu_{X_0}\lR_2 = \sigma^2_{X_0}$ and $\lL \mu_W \lR = 0, \lL \mu_W\lR_2 = \sigma^2_W$.
The empirical distributions $\widehat{\bx}_0 \xRightarrow{\text{d}} \mu_{X_0}$ as $N \rightarrow \infty$, $\widehat{\bw} \xRightarrow{\text{d}} \mu_{W}$ as $n \rightarrow \infty$, and  $\mu_{X_0}f < \infty$ and $\mu_{W}f < \infty$ for any controlled function $f$.  

\subsection{Convergence Lemmas}
\label{MCR}
To find an asymptotic expression SE for the measurement ensemble $\cI$, we establish the following lemmas.
\begin{proposition}
\label{ps5}
Let $A(n)$ be a random variable with $A(n) \sim \frac{1}{\sqrt{n}}\mu$, where $\langle \mu \rangle = 0$ and $\langle\mu \rangle_2 = 1$.
Then \small$\mathbb{E}[A^{2+2\alpha}(n)] = o(n^{-2})$\normalsize~for $\alpha >1$.   
\end{proposition}
\begin{proof}
See Appendix~\ref{ProofPS4}.
\end{proof} 
    \noindent We now present a the key proposition in our SE analysis.  
    \begin{proposition}\label{ps6}
    We let \small$\bA(N) \in \mathbb{R}^{n \times N}$\normalsize~be a random matrix defined as in Section~\ref{definitionA}.
    Suppose that  \small$\mathbf{v}(N) \in \mathbb{R}^{N \times 1}$\normalsize~is a deterministic vector satisfying \small$\lim_{N \rightarrow \infty} \langle \mathbf{v}^2(N) \rangle = s^2_0 < \infty$\normalsize~and \small$\limsupN  \frac{1}{N}\|\bv(N)\|^{2+2\alpha}_{2+2\alpha} < \infty$\normalsize~for a constant $\alpha >1$.
    Then, the following empirical law converges, \small$\widehat{\bA(N)\bv(N)} \xRightarrow{\text{d}} \mathcal{N}\left(0,{s^2_0}/{\rho}\right)$ as $N \rightarrow \infty$\normalsize~while $\rho = \frac{n}{N}$ is kept being constant.  
    \end{proposition}
    \begin{proof}   
    See Appendix~\ref{ProofPS6}.
    \end{proof}
   \noindent The next proposition is a direct consequence  of {Proposition}~\ref{ps6}.
   \begin{proposition}
   \label{ps7}
   Suppose $\bu \in \R^{N \times 1}$ and $\bv \in \R^{n \times 1}$ with \small$\|\bu\|_2 = \|\bv\|_2 = 1$. 
   Then $\bv^*\bA\bu \xRightarrow{d}  \frac{Z}{\sqrt{n}}$ as $N \rightarrow \infty$, where the matrix $\bA$ is defined as in {Proposition}~\ref{ps6} and $Z \sim \cN(0,1)$. 
   \end{proposition}

\begin{remark}\label{rmk2}
It is worth noting that Proposition~\ref{ps7} is a generalization of  Lemma 2 in \cite{Bayati11}. In \cite{Bayati11}, the standard properties of Gaussian matrices are used to show Proposition~\ref{ps7}. 
However, this can not be directly extended to the measurement matrices in $\cI$ because its entries are not i.i.d.
\vspace{-0.2cm}
\end{remark}
\subsection{State Evolution of General AMP Algorithm}
     Suppose that 
     \vspace{-0.0cm}
     \small\begin{equation}
     \label{eq:q0}
      \lim_{N \rightarrow \infty} \langle \bq^0, \bq^0 \rangle/\rho = \sigma_0^2 < \infty,~ \limsupN 
 \|\bq^0\|^{2+2\alpha}_{2+2\alpha}/N < \infty,
     \vspace{-0.0cm}
     \end{equation}\normalsize
     for a constant $\alpha > 1$. 
     Then, the SE for the general AMP in \eqref{eq:6} is described \cite{Bayati11} by 
\vspace{-0.0cm}
\small\begin{equation}
	\label{eq:SE}
	\tau_t^2 = \mathbb{E}\left[g_t^2(\sigma_t Z,W)\right],~~ \sigma_t^2 = \mathbb{E}\left[f_{t}^2(\tau_{t-1}Z,X_0)\right]/\rho,
 \vspace{-0.0cm}
\end{equation}\normalsize 
where \small$X_0 \sim \mu_{X_0}, W \sim \mu_{W}$\normalsize, and \small$Z \sim \mathcal{N}(0,1)$\normalsize~is independent of $W$ and $X_0$. 
The SE in \eqref{eq:SE} has the same expression as the one in \cite{Bayati11, Rush2018} except that the matrix $\bA$ in our setting follows the new assumptions in Section~\ref{definitionA} and the controlled function. 

Define a collection of vectors as a set \sloppy \small$\cF_{t_1,t_2} = \{\bb^0,\ldots, \bb^{t_1\-1}, \bm^0,\dots,\bm^{t_1\-1}, \bh^1,\dots,\bh^{t_2},\bq^0,\dots,\bq^{t_2},\bx_0, \bw\}$\normalsize. 
The recursion in \eqref{eq:6} can then be represented by incorporating matrices as \small$\bY_t = \bA \bQ_t$\normalsize~with \small$\bY_t = [\bb^0|{\bb^{1} + \lambda_{1}\bm^{0}}|\dots|\bb^{t-1} + \lambda_{t-1}\bm^{t-2}] \in \R^{n \times t}$\normalsize~and \small$\bQ_t = [\bq^0|\bq^1| \dots|\bq^{t-1}] \in \R^{N \times t}$\normalsize, and \small$\bX_t = \bA^*\bM_t$\normalsize~with \small$\bX_t = [\bh^1+\xi_0\bq^0|\dots| \bh^t +\xi_{t-1}\bq^{t-1}] \in \R^{N \times t}$\normalsize~and \small$\bM_t = [\bm^0|\dots|\bm^{t-1}] \in \R^{n \times t}$\normalsize. 
The $\bm^t_{||}$ and $\bq^t_{||}$ denote, respectively, the projection of $\bm^t$ and $\bq^t$ onto the column spaces of $\bM_t$ and $\bQ_t$,
\vspace{-0.1cm}
\small\begin{equation}\label{eq:projection1}
   \bm^t_{||} = {\sum_{i=0}^{t-1}\zeta_i\bm^i},~\bq^t_{||} = \sum_{i=0}^{t-1}\beta_i\bq^i,
\vspace{-0.1cm}
\end{equation}\normalsize
\noindent where $\zeta_i$ and $\beta_i$ quantify, respectively, the contributions of $\bm^i$ and $\bq^i$ to the projected images of $\bm^t$ and $\bq^t$.
Then their null space projections 
$\bq^t_{\perp}$ and $\bm^t_{\perp}$ are defined as 
\vspace{-0.1cm}
\small\begin{equation}\label{eq:projection2}
  { \bq^t_{\perp} = \bq^t - \bq^t_{||}},~ \bm^t_{\perp} = \bm^t - \bm^t_{||}.
   \vspace{-0.1cm}
\end{equation}\normalsize
The following theorem is the main result of this work.  
\begin{theorem} \label{theorem1}
Suppose $\{\tau_t\}_{t\geq 0}$ and $\{\sigma_t\}_{t\geq 0}$ defined in \eqref{eq:SE}. 
Given the AMP recursion in \eqref{eq:6}, the following propositions hold for $t \geq 0$. \\
    a)
    \vspace{-0.3cm}
    \small\begin{equation}
    \label{eq:b1}
    \bb^t|_{\cF_{t,t}} \xRightarrow{\text{d}} \sum_{j=0}^{t-1}\beta_j\bb^j +\widetilde{\bA}\bq^{t}_{\perp}, \text{ as } n \rightarrow \infty,	\end{equation}\normalsize
    \vspace{-0.2cm}
    \small\begin{equation}
    \label{eq:h1}
    \bh^{t+1}|_{\cF_{t+1,t}} \xRightarrow{\text{d}} \sum_{j=0}^{t-1}{\zeta_j}\bh^{j+1} +\widetilde{\bA}^*\bm^t_{\perp}, \text{ as } N \rightarrow \infty,
    \vspace{-0.2cm}
    \end{equation}\normalsize
    where $\widetilde{\bA} \ddis \bA$.\\
    \vspace{-0.3cm}
    b) For a constant $\alpha > 1$, 
    \vspace{0.2cm}
    \small\begin{equation}
    \label{eq:alphaQ}    
    \limsupN \frac{1}{N} \|\bq^t_{\perp}\|^{2+2\alpha}_{2+2\alpha} < \infty, 
    \end{equation}
    \vspace{-0.1cm}
    \begin{equation}
    \label{eq:alphaM}    
    \limsupn \frac{1}{n} \|\bm^t_{\perp}\|^{2+2\alpha}_{2+2\alpha} < \infty.
    \vspace{-0.2cm}
    \end{equation}\normalsize\\  
    c) For all $0 \leq t_1,t_2 \leq t$, 
   \vspace{0.1cm}
    \small\begin{equation}
    \label{eq:b2}
    \limn \lL \bb^{t_1},\bb^{t_2} \lR \das \frac{1}{\rho} \limN \lL \bq^{t_1},\bq^{t_2}\lR < \infty.
    \end{equation}
    \vspace{-0.2cm}
    \begin{equation}
    \label{eq:h2}
    \limN \lL\bh^{t_1+1},\bh^{t_2+1} \lR \das \limn \lL \bm^{t_1},\bm^{t_2}\lR < \infty.
    \end{equation}\normalsize\\
    \vspace{-0.2cm}
    d) For all controlled functions $ \phi_b$ and $\phi_h$: $\mathbb{R}^{t+2} \rightarrow \mathbb{R}$,  
    \small
    \begin{equation}
    \label{b:b} 
    \limn  \frac{1}{n} \sum_{i=1}^{n}\phi_b(\bu^t_i) 	\das  \mathbb{E}\left[\phi_b(\sigma_0\widetilde{Z}_0,.., \sigma_t\widetilde{Z}_t, W)\right],
    \end{equation}	
    \vspace{-0.2cm}
    \begin{equation}
    \label{b:h} 
	 \limN \frac{1}{N}\sum_{i=1}^{N}\phi_h(\bv^t_i) 
	 \das 
	 \mathbb{E}\left[\phi_h(\tau_0Z_0, ..,\tau_tZ_t, X_0)\right], 
    \vspace{-0.2cm}
    \end{equation}\normalsize
    where $\bu^t_i = (b^0_i \dots, b^{t}_i,w_i)$ and $\bv^t_i = (h^1_i,.., h^{t\+1}_i,x_{0i})$. 
    The $(Z_0,\dots, Z_t)$ and $(\widetilde{Z}_0, \dots, \widetilde{Z}_t)$ are Gaussian vectors, where $Z_j$ and $\widetilde{Z}_j$ are i.i.d. following $\cN(0,1)$, for $j=1,2,\dots, t$, and independent of $X_0$ and $W$.
\end{theorem}
\begin{remark}
The results in \eqref{b:b} and \eqref{b:h} are similar to Lemma 1b in \cite{Bayati11} except that $\phi_b$ and $\phi_h$ in \eqref{b:b} and \eqref{b:h} are controlled functions.
The results in \eqref{eq:b1} and \eqref{eq:h1} represent  the convergence in the distribution of $\bb^t|_{\cF_{t,t}}$ and $\bh^{t+1}|_{\cF_{t+1,t}}$, which are obtained by applying the conditioning technique and Proposition~\ref{ps4} in our analysis. 
Compared to Lemma 1a in \cite{Bayati11}, the expressions in  \eqref{eq:b1} and \eqref{eq:h1} do not include the basis-aligned deviation terms. 
More specifically, the $\bb^t|_{\cF_{t,t}}$ and $\bh^{t+1}|_{\cF_{t+1,t}}$ in Lemma 1a of \cite{Bayati11} are 
\vspace{-0.2cm}
\small\begin{equation}\label{eq:b1B}
\bb^t|_{\cF_{t,t}} \ddis \sum_{j=0}^{t-1}\beta_j\bb^j +\widetilde{\bA}\bq^{t}_{\perp}+\widetilde{\bM}_t\overrightarrow{\bo}_t(1),\end{equation}\normalsize
\vspace{-0.2cm}
\small\begin{equation}
\label{eq:h1B}
\bh^{t+1}|_{\cF_{t+1,t}} \ddis \sum_{j=0}^{t-1}{\zeta_j}\bh^{j+1} +\widetilde{\bA}^*\bm^t_{\perp}+ \widetilde{\bQ}_{t+1}\overrightarrow{\bo}_t(1),
\vspace{-0.2cm}
 \end{equation}\normalsize
where \small$\widetilde{\bM}_t \in \R^{n \times t}$\normalsize~and \small$\widetilde{\bQ}_{t+1} \in \R^{N \times (t+1)}$\normalsize~are the  orthogonal bases of the column subspaces of $\bM_t$ and {$\bQ_{t+1}$}, respectively.
The $\overrightarrow{\bo}_t(1)$ denotes a vector in $\R^{t \times 1}$ such that all entries converge to $0$ almost surely as $N \rightarrow \infty$.
The common procedure for characterizing the SE in  \cite{Bayati11,Rush2018,Takeuchi2020,Rangan2019,Zhou2020} is to cancel the deviation terms \small$\widetilde{\bM}_t\overrightarrow{\bo}_t(1)$\normalsize~and \small$\widetilde{\bQ}_{t+1}\overrightarrow{\bo}_t(1)$\normalsize~in \eqref{eq:b1B} and \eqref{eq:h1B}, respectively, by assuming the smoothness of the pseudo-Lipschitz functions.    
The challenge in our proof of {Theorem}~\ref{theorem1} is that there is no such smoothness for $\phi_b$ and $\phi_h$ in \eqref{b:b} and \eqref{b:h}, respectively, because they are controlled functions.
It is shown in Section~\ref{secIV} that these deviation terms vanish in our derivation, which leverages the concept of empirical statistics developed in Section~\ref{secIIA}.   
We note that the SE in \eqref{eq:SE} is a special case of Theorem~\ref{theorem1}; its detailed derivation is relegated to Appendix~\ref{AppendixA}.
\end{remark}
\begin{remark}
    The results in Theorem~\ref{theorem1} directly verify the conjecture about the validity of SE for i.i.d. non-Gaussian measurement ensemble in \cite{Bayati11} because the set of such $\bA$ is a subset of the $\bA$ defined in Section~\ref{definitionA}.  
\end{remark}
\section{Proof of Theorem 1}
\label{secIV}

 The proof of {Theorem}~\ref{theorem1} is inspired by  \cite{Bayati11}.
 Since a part of the proof is based on a similar technique in \cite{Bayati11}, we refer to \cite{Bayati11} for those standard arguments, while we present the features that are unique and refined in this work.
The proof is based on mathematical induction in $t$ and the application of the conditioning technique on four sets $\cF_{0,0},\cF_{1,0},\cF_{t,t}$, and $\cF_{t+1,t}$ sequentially. 
\subsection{Step 1: Conditioning on $\cF_{0,0} = \{\bq^0,\bx_0,\bw\}.$}

a)  The convergence in  \eqref{eq:b1} holds because $\bb^0 = \bA\bq^0 \ddis \widetilde{\bA}\bq_{\perp}^0$, where $\bq_{\perp}^{0} = \bq^0$ because $\bQ_0$ is an empty matrix. 

b) The bound in \eqref{eq:alphaQ} is due to \eqref{eq:q0} and $\bq^0_{\perp} = \bq^0$.

c)
Given \eqref{eq:q0}, applying {Proposition}~\ref{ps6} to $\bA$ and $\bq^0$ leads to 
 \vspace{-0.2cm}
\small \begin{equation}
     \label{eq:1.1}
     \widehat{\bA \bq^0} \xRightarrow{\text{d}} \cN(0,\sigma^2_0).
 \vspace{-0.2cm}
 \end{equation}\normalsize
 Hence, by Step 1a), \small$\limn \lL \bb^0, \bb^0\lR \das \sigma^2_0$\normalsize, implying  \small$\limn \lL \bb^0, \bb^0 \lR \das \limN \frac{\lL \bq^0,\bq^0\lR}{\rho} < \infty$\normalsize, which completes the proof of \eqref{eq:b2}. 
d) We let $\widetilde{\bu}^0_i=(\sigma_0 \widetilde{Z}_0,w_i)$, $\forall i$, where $\widetilde{Z}_0 \sim \cN(0,1)$.
From \eqref{eq:1.1}, the convergence  $\bu^0_i \xRightarrow{\text{d}}  \widetilde{\bu}^0_i$ holds.
To prove \eqref{b:b}, we first claim that 
\small$\limn \frac{1}{n}\sumn \phi_b(\widetilde{\bu}^0_i) \das \E[\phi_b(\sigma_0\widetilde{Z}_0,W)]$\normalsize.
By the triangular inequality, we have 
\small$\Big| \frac{1}{n}\sumn \phi_b(\widetilde{\bu}^0_i) - \E[\phi_b(\sigma_0\widetilde{Z}_0,W)]\Big| \leq X^0_1 + X^0_2$\normalsize,
 where \small$X^0_1 = \Big|\frac{1}{n}\sumn \Big(\phi_b(\widetilde{\bu}^0_i) - \E_{\widetilde{Z}_0}[\phi_b(\widetilde{\bu}^0_i)]\Big)  \Big|$\normalsize~and \small$X^0_2 = \Big|\frac{1}{n}\sumn \E_{\widetilde{Z}_0}[\phi_b(\widetilde{\bu}^0_i)] - \E[\phi_b(\sigma_0\widetilde{Z}_0,W)] \Big|$\normalsize. 
 The goal is to prove that \small$\limn X^0_1 \das 0$\normalsize~and \small$\limn X^0_2 \das 0$\normalsize, respectively. 
Since $\phi_b$ is a controlled function, we get by \eqref{eq:controlled2}  \small$|\phi_b(\widetilde{\bu}^0_i)| \leq c^0_1 \exp(c^0_2(|\sigma_0\widetilde{Z}_0|^{\lambda} + |w_i|^{\lambda}))$\normalsize, where $c^0_1>0$, $c^0_2>0$, and $1 \leq \lambda <2$ are constant. 
Hence, $\E_{\widetilde{Z}_0}[|\phi_b(\widetilde{\bu}^0_i)|^{2+\kappa}] \leq c^0_3 \exp(c^0_4|w_i|^{\lambda}) \E_{\widetilde{Z}_0}[\exp(c^0_4|\sigma_0\widetilde{Z}_0|^{\lambda})]$, where $0<\kappa <1$, $c^0_3 = (c^0_1)^{2+\kappa}$, and $c^0_4 = c^0_2(2+\kappa)$.
Thus, 
\begin{equation}
    \label{eq:1.2}
    \E_{\widetilde{Z}_0}[|\phi_b(\widetilde{\bu}^0_i)|^{2+\kappa}] \leq c^0_5 \exp(c^0_4|w_i|^{\lambda}),
\end{equation}
where $c^0_5 = c^0_3 \E_{\widetilde{Z}_0}[\exp(c^0_4|\sigma_0\widetilde{Z}_0|^{\lambda})]$ is constant.
Define $X^0_{n,i} = \phi_b(\widetilde{\bu}^0_i) - \E_{\widetilde{Z}_0}[\phi_b(\widetilde{\bu}^0_i)]$, $\forall i$. 
To prove the thesis $\limn X^0_1 \das0$, we show that $\{X^0_{n,i}\}^n_{i=1}$ satisfy {Lemma}~\ref{lm3} in Appendix~\ref{AppLM}.
Indeed, applying Holder's inequality in Lemma~\ref{lm4} in Appendix~\ref{AppLM} to $X^0_{n,i}$ gives 
\vspace{-0.2cm}
\small\begin{subequations}
\label{eq:1.3.1}
\beq 
\E_{\widetilde{Z}_0}[|X^0_{n,i}|^{2+\kappa}] \d4&\leq&\d4 2^{1+\kappa} \left( \E_{\widetilde{Z}_0}[|\phi_b(\widetilde{\bu}^0_i)|^{2+\kappa}]  +  |\E_{\widetilde{Z}_0}[\phi_b(\widetilde{\bu}^0_i)]|^{2+\kappa} \right), \nonumber\\ 
\d4&\leq&\d4 c^0_6\exp(c^0_4|w_i|^{\lambda}),\label{eq:1.3e} 
\eeq 
\end{subequations}\normalsize
where \small$c^0_6 = 2^{2+\kappa}c^0_5$\normalsize~and \eqref{eq:1.3e} follows from Lemma~\ref{lm5} (Lyapunov's inequality) in Appendix~\ref{AppLM} and \eqref{eq:1.2}. 
Note that \small$\frac{1}{n} \sumn c^0_6\exp(c^0_4|w_i|^{\lambda})  
  \das \E[c^0_6\exp(c^0_4|W|^{\lambda})]$\normalsize. 
  For $n$ sufficiently large, using \eqref{eq:1.3e} gives      
\vspace{-0.2cm}\small\begin{equation}
  \frac{1}{n}\sumn \E_{\widetilde{Z}_0}[|X^0_{n,i}|^{2+\kappa}]  \leq \E[c^0_6\exp(c^0_4|W|^{\lambda})] < cn^{\kappa/2},\label{eq:1.4b}  
\vspace{-0.2cm}
\end{equation}\normalsize 
 where $c > 0$ is constant and \eqref{eq:1.4b} follows from the facts that $\E[c^0_6\exp(c^0_4|W|^{\lambda})]=c^0_7< \infty$ and there exists a constant $n_0 >0$ such that ${c^0_7}<cn^{\kappa/2}$ for $n >n_0$.
{Lemma}~\ref{lm3} in Appendix~\ref{AppLM} leads to $\limn \frac{1}{n}\sumn X^0_{n,i} \das 0$, implying $\limn X^0_1 \das 0$.   
 The rest of {Step}~1d) is showing the convergence \small$\limn X^0_2 \das0$\normalsize. 
 Define \small$\widetilde{\phi}_b(w_i) = \E_{\widetilde{Z}_0}[\phi_b(\widetilde{\bu}^0_i)]$\normalsize, $\forall i$.
 Then \small$\limn \frac{1}{n}\sumn \widetilde{\phi}_b(w_i) \das \E[\widetilde{\phi}_b(W)]$\normalsize~because $\widehat{\bw} \xRightarrow{\text{d}} \mu_W$.
 Hence, \small$\limn \frac{1}{n} \sumn \E_{\widetilde{Z}_0}[\phi_b(\widetilde{\bu}^0_i)] \das \E_{W}[\E_{\widetilde{Z}_0}[\phi_b(\sigma_0\widetilde{Z}_0,W)]]\das \E[\phi_b(\sigma_0\widetilde{Z}_0,W)]$\normalsize,
implying $\limn X^0_2 \das 0$, concludes \eqref{b:b}.  
 \subsection{Step 2: Conditioning on \small$\cF_{1,0} = \{\bb^0,\bm^0,\bq^0,\bx_0,\bw\}$\normalsize.}
a) From \eqref{eq:6}, $\bh^1|_{\cF_{1,0}} \ddis \bA^*|_{\cF_{1,0}}\bm^0 - \xi_0\bq^{0}$.
Conditioning on $\cF_{1,0}$ is equivalent to conditioning on $\{\bA|\bA\bq^0=\bb^0 \}$. 
Hence, applying {Proposition}~\ref{ps4} yields   \small$\bA|_{\cF_{1,0}} \ddis \widetilde{\bA}\bP^{\perp}_{\bq^{0}} +\frac{\bb^0\bq^{0*}}{\|\bq^0\|^2_2}$\normalsize.
Then $\bh^1|_{\cF_{1,0}} \ddis \bP^{\perp}_{\bq^0}\widetilde{\bA}^*\bm^0 + \frac{\bb^{0*}\bm^0}{\|\bq^0\|^2_2}\bq^0 - \xi_0\bq^0 $, resulting in  
\vspace{-0.2cm}
\small\begin{equation}\label{eq:2.1}
    \bh^1|_{\cF_{1,0}} \das  \widetilde{\bA}^*\bm^0-\bP_{\bq^0} \widetilde{\bA}^*\bm^0 +\Big(\rho \frac{\lL\bb^0,\bm^0\lR}{\lL\bq^0,\bq^0\lR}-\xi_0\Big)\bq^0.
    \vspace{-0.2cm}
\end{equation}\normalsize
Substituting \small$\bm^0 =g_0(\bb^0,\bw)$\normalsize~into \small$\lL \bb^0,\bm^0 \lR$\normalsize~in \eqref{eq:2.1} gives 
\small$\limn \lL \bb^0,\bm^0 \lR \das \E[\sigma_0 \widetilde{Z}_0g_0(\sigma_0 \widetilde{Z}_0,W)] \das \sigma^2_0 \E[g'_0(\sigma_0\widetilde{Z},W)] \das \limn \lL \bb^0,\bb^0\lR \lL g_0'(\bb^0,\bw) \lR$\normalsize, 
where the first equality is due to \eqref{b:b} applied to \small$\phi_b(\bu^0_i) = b^0_ig_0(b^0_i,w_i)$\normalsize, the second equality is due to Stein's Lemma (Lemma~\ref{lm2}) in Appendix~\ref{AppLM}, the third equality is obtained by setting \small$\phi_b(\bu^0_i) = g'_0(b^0_i,w_i)$\normalsize~in \eqref{b:b}.
Note that $\xi_0 = \lL g'_0(\bb^0,\bw)\lR$. Using \eqref{eq:b2} gives \small$\limn \lL \bb^0,\bm^0 \lR \das \limN \frac{\lL\bq^0,\bq^0 \lR}{\rho}\xi_0$\normalsize, which is plugged into   \eqref{eq:2.1} to result in 
    \vspace{-0.2cm}
    \small\begin{equation}
        \label{eq:2.3}
        \bh^1|_{\cF_{1,0}} \ddis  \widetilde{\bA}^*\bm^0-\bP_{\bq^0} \widetilde{\bA}^*\bm^0 +o(1)\bq^0,
    \vspace{-0.2cm}
    \end{equation}\normalsize
    where $\limn o(1) \das 0$. Using {Proposition}~\ref{ps3} in Appendix~\ref{AppLM}, the second term on the right-hand side {(r.h.s)} of \eqref{eq:2.3} converges to \small$\limN \bP_{\bq^0} \widetilde{\bA}^*\bm^0 \das \mathbf{0}_N$\normalsize. 
   We claim that the last term converges to $\limN o(1)\bq^0 \das \mathbf{0}_N$. Indeed, this is true because (i) the empirical expectation of $o(1)\bq^0$ satisfies \small$\limN \lL \widehat{o(1)\bq^0} \lR \das 0$\normalsize, which holds due to  $\limN |\lL\widehat{o(1)\bq^0}\lR| \leq \limN |o(1)| \frac{1}{N}\sumN |q^0_i| \das 0$ and (ii) the empirical variance of $o(1)\bq^0$ converges to $\limN \lL \widehat{o(1)\bq^0} \lR_2 = \limN [o(1)]^2\lL\bq^0,\bq^0 \lR \das 0$. 
    Thus,  $\bh^1|_{\cF_{1,0}} \xRightarrow{\text{d}} \widetilde{\bA}^*\bm^0$, which concludes \eqref{eq:h1} when $t = 0$ since $\bm^0 = \bm^0_{\perp}$; note that $\bM_0$ is an empty matrix.
    
    b) The bound in \eqref{eq:alphaM} is equivalent to \small$\limsupn \frac{1}{n} \sumn |m^0_i|^{2+2\alpha} < \infty$\normalsize. 
{    
Incorporating \small$\phi_b(b^0_i,w_i) = |g_0(b^0_i,w_i)|^{2+2\alpha}$\normalsize~for a constant $\alpha >1$ into \eqref{b:b} leads to $\limn \frac{1}{n} \sumn |m^0_i|^{2+2\alpha} \das \E[|g_0(\sigma_0\widetilde{Z}_0,W)|^{2+2\alpha}] < \infty$, which completes the proof.
}
    c) By \eqref{b:b} at $t=0$, we get  \small$\limn \lL g_0(\bb^0,\bw),g_0(\bb^0,\bw)\lR\das \E[g^2_0(\sigma_0\widetilde{Z}_0,W)]$\normalsize,~leading to $\limn \lL \bm^0,\bm^0\lR \das \tau^2_0$ due to the definitions of $\bm^0$ in \eqref{eq:6} and $\tau^2_0$ in \eqref{eq:SE}.
    Thus, {Proposition}~\ref{ps6} (Lindeberg-Feller) holds because $\bm_{\perp}^0 = \bm^0$ and \eqref{eq:alphaM}, resulting in \small$\widetilde{\bA}^*\bm^0 \xRightarrow{\text{d}} \cN(0,\tau^2_0)$\normalsize~and \small$\widehat{\bh^1}\xRightarrow{\text{d}} \cN(0,\tau^2_0)$\normalsize, i.e., \small$\limN \lL\bh^1,\bh^1\lR \das \tau^2_0$\normalsize. 
    Hence, \small$\limN \lL\bh^1,\bh^1\lR \!\das\! \limn \lL\bm^0,\bm^0 \lR \!<\! \infty$\normalsize, concluding \eqref{eq:h2}.
    d) We let $\widetilde{\bv}^0_i=(\tau_0Z_0,x_{0i})$, $\forall i$. 
    Because $\bv_i^0 \xRightarrow{\text{d}} \widetilde{\bv}^0_i$ holds due to \small$\widehat{\bh^1_i}\xRightarrow{\text{d}} \cN(0,\tau^2_0)$\normalsize~(Lindeberg-Feller), $\forall i$, 
    the proof of \eqref{b:h} is boiled down to showing \small$\limN \frac{1}{N}\sumN \phi_h(\widetilde{\bv}^0_i) \das \E[\phi_h(\tau_0Z_0,X_0)]$\normalsize. 
    Similar to {Step} 1d), by the triangular inequality \small$\Big|\frac{1}{N}\sumN \phi_h(\widetilde{\bv}^0_i) - \E[\phi_h(\tau_0Z_0,X_0)]\Big|\leq Y^0_1+Y^0_2$\normalsize, where \small$Y^0_1 = \Big|\frac{1}{N}\sumN \phi_h(\widetilde{\bv}^0_i) - \E_{Z_0}[\phi_h(\widetilde{\bv}^0_i)]\Big|$\normalsize~and  
    \small$Y^0_2 = \Big|\frac{1}{N} \sumN \E_{Z_0}[\phi_h(\widetilde{\bv}^0_i)] - \E[\phi_h(\tau_0Z_0,X_0)]\Big|$\normalsize.
    We claim that $\limN Y^0_1 \das 0$ and $\limN Y^0_2 \das 0$.
    The convergence $\limN Y^0_1 \das 0$ is treated first.
    Defining $Y^0_{N,i} = \phi_h(\widetilde{\bv}^0_i) - \E_{Z_0}[\phi_h(\widetilde{\bv}^0_i)]$, the proof is equivalent to showing that $\{Y^0_{N,i}\}_{i=1}^{N}$ satisfy {Lemma}~\ref{lm3} in Appendix~\ref{AppLM}.
    Since $\phi_h$ is a controlled function, the following holds \small$|\phi_h(\widetilde{\bv}^0_i)| \leq d^0_1 \exp(d^0_2(|\tau_0Z_0|^{\lambda} + |x_{0i}|^{\lambda}))$\normalsize, where $d^0_1>0$, $d^0_2>0$, and $1 \leq  \lambda < 2$ are constant.
    Hence, \small$\E_{Z_0}[|\phi_h(\widetilde{\bv}^0_i)|^{2+\kappa}] \leq d^0_3\exp(d^0_4|x_{0i}|^{\lambda})\E_{Z_0}[\exp(d^0_4|\tau_0Z_0|^{\lambda})]$\normalsize, where \small$0 < \kappa < 1$, $d^0_3 = (d^0_1)^{2+\kappa},$ and $d^0_4 = d^0_2(2+\kappa)$\normalsize.
    Therefore, 
    \small\begin{equation}
    \label{eq:2.4}
    \E_{Z_0}[|\phi_h(\widetilde{\bv}^0_i)|^{2+\kappa}] \leq d^0_5\exp(d^0_4|x_{0i}|^{\lambda}),
    \end{equation}\normalsize
    where \small$d^0_5 = d^0_3\E_{Z_0}[\exp(d^0_4|\tau_0Z_0|^{\lambda})]$\normalsize~ is a constant. 
    By {Lemma}~\ref{lm4} in Appendix~\ref{AppLM},
    \vspace{-0.2cm}
    \small\begin{subequations}
    \label{eq:2.5}
    \beq 
        \E_{Z_0}[|Y^0_{N,i}|^{2+\kappa}]\d4 &\leq&\d4 2^{1+\kappa}\left(\E_{Z_0}[|\phi_h(\widetilde{\bv}^0_i)|^{2+\kappa}] + |\E_{Z_0}[\phi_h(\widetilde{\bv}^0_i)]|^{2+\kappa} \right),\nonumber\\
        \d4&\leq&\d4 d^0_6 \exp(d^0_4|x_{0i}|^{\lambda})\label{eq:2.5e},
    \eeq
    \end{subequations}\normalsize 
    where $d^0_6 =2^{2+\kappa}d^0_5$ and  \eqref{eq:2.5e} follows from Lemma~\ref{lm5} in Appendix~\ref{AppLM} and \eqref{eq:2.4}.
    Note that $\frac{1}{N}\sumN d^0_6 \exp(d^0_4|x_{0i}|^{\lambda}) \das \E[d^0_6\exp(d^0_4|X_0|^{\lambda})]$. 
    For $N$ is sufficiently large,  we have 
    \vspace{-0.2cm}
    \small\begin{equation}
        \d4 \frac{1}{N}\sumN \E_{Z_0}[|Y^0_{N,i}|^{2+\kappa}] \leq \E[d^0_6\exp(d^0_4|X_0|^{\lambda})]<  cN^{\kappa/2}, \label{eq:2.6c}
    \vspace{-0.2cm}
    \end{equation}\normalsize 
    where $c>0$ is constant and the last estimation in \eqref{eq:2.6c} follows because  $\E[d^0_6\exp(d^0_4|X_0|^{\lambda})] = d^0_7 < \infty$ and there exists a constant $N_0 > 0$ such that $d^0_7 < cN^{\kappa/2}$ for $N > N_0$.
    Given \eqref{eq:2.6c}, applying {Lemma}~\ref{lm3} in Appendix~\ref{AppLM} to $\{Y^0_{N,i}\}_{i=1}^{N}$ gives $\limN \frac{1}{N}\sumN Y^0_{N,i} \das 0$, implying $\limN Y^0_1 \das 0$. 
    To show \small$\limN Y^0_{2}\das 0$, we define $\widetilde{\phi}_h(x_{0i}) = \E_{Z_0}[\phi_h(\widetilde{\bv}^0_i)]$\normalsize. 
    Then the following holds,
    \small$\limN \frac{1}{N} \sumN \widetilde{\phi}_h(x_{0i}) \das \E[\widetilde{\phi}_h(X_0)]$\normalsize~because \small$\widehat{\bx_0} \xRightarrow{\text{d}} \mu_{X_0}$\normalsize.
    Thus, \small$\limN \frac{1}{N} \sumN \E_{Z_0}[\phi_h(\widetilde{\bv}^0_i)] \das \E_{X_0}[\E_{Z_0}[\phi_h(\tau_0Z_0,X_0)]]=\E[\phi_h(\tau_0Z_0,X_0)]$\normalsize, resulting in $\limN Y^0_2 \das 0$, concluding \eqref{b:h}.   
    Suppose {Theorem}~\ref{theorem1} holds up to the $(t-1)$th iteration. 
    We prove that the thesis also holds for the $t$th iteration.
    Similar to the first two steps, we show \eqref{eq:b1}, \eqref{eq:alphaQ}, \eqref{eq:b2}, and \eqref{b:b} in {Step}~3 and \eqref{eq:h1}, \eqref{eq:alphaM}, \eqref{eq:h2},  and \eqref{b:h} in {Step}~4, respectively, which are more complex and thus, relegated to Appendix~\ref{AppendixC}. 

\section{Conclusion} 
The SE analysis for AMP was extended to the  class of measurement matrices with independent (not necessarily identically distributed) entries and the controlled functions. 
A variant of the Lindeberg-Feller was proposed to deal with the lack of the structure of the assumed measurement ensemble.
An empirical statistic-based conditioning technique was proposed to cope with the lack of smoothness of the controlled functions. 
The results revealed a new direction to the SE analysis for boarder classes of measurement ensembles and functions.   

\bibliographystyle{IEEEtran}
\bibliography{biblib}
\newpage
\appendices
\section{Pseudo-Lipschitz Function}\label{PLdef}
\begin{definition}
For a $k > 1$, a function $f:\R^{n \times 1} \rightarrow \R$ is said pseudo-Lipschitz of order $k$ if there exists a constant $L > 0$ such that $|f(\bx)-f(\by)| \leq L(1+\|\bx\|^{k-1}+\|\by\|^{k-1})\|\bx-\by\|$, for all $\bx, \by \in \R^{n \times 1}$;  the first order derivative of $f$ is bounded by a polynomial of order $(k-1)$, i.e., polynomial smoothness.
\end{definition}
\section{Proof of Proposition~\ref{ps1}}\label{ProofPS1}
    This follows from \small$\sum_{m=1}^{n}\mathbb{E}\big[ X^2_{n,m};|X_{n,m}| > \epsilon \big] \leq \sum_{m=1}^{n}\mathbb{E} \Big[\frac{X^{2+2\alpha}_{n,m}}{\epsilon^{2\alpha}} \Big]$\normalsize, where the right-hand side converges to zero because $\frac{n}{\epsilon^{2\alpha}}o(n^{-1}) \rightarrow 0$ as $n \rightarrow \infty$.
\section{ Proof of Proposition~\ref{ps4}}\label{ProofPS3}
 Recalling the following subspace decomposition \small$\bA = \bP^{\perp}_{\bM}\bA\bP^{\perp}_{\bQ} + \bA\bP_{\bQ} + \bP_{\bM}\bA - \bP_{\bM}\bA\bP_{\bQ}$\normalsize, the orthogonal projection of $\bA$ onto  $\cP_{\bX,\bY}$ is given by \small$\bA|_{\cP_{\bX,\bY}} = \bP^{\perp}_{\bM}\bA\bP^{\perp}_{\bQ} + \bB$\normalsize. 
    For any integrable function $\psi$, \small$\E\left[\psi(\bA|_{\cP_{\bX,\bY}})\right] = \E \left[\psi(\bP^{\perp}_{\bM}\bA\bP^{\perp}_{\bQ} + \bB) \right] =  \E \left[\psi(\bP^{\perp}_{\bM}\widetilde{\bA}\bP^{\perp}_{\bQ} + \bB) \right] =\E \left[\psi(\widetilde{\bA}|_{\cP_{\bX,\bY}}) \right]$\normalsize, where the second equality follows from the fact that \small$\bA \stackrel{d}{=}\widetilde{\bA}$\normalsize. 
    Hence, \small$ \bA|_{\cP_{\bX,\bY}} \stackrel{d}{=} \bP^{\perp}_{\bM}\widetilde{\bA}\bP^{\perp}_{\bQ} + \bB$\normalsize, which completes the proof.
\section{Proof of Proposition~\ref{ps5}}\label{ProofPS4}
    Denoting $B = \sqrt{n}A(n) \sim \mu$ yields $\mathbb{E}\big[A^{2+2\alpha}(n)\big] = \mathbb{E}\big[B^{2(1+\alpha)}\big]n^{-(1+\alpha)}= o(n^{-2})$, where the last step uses the facts that $\mathbb{E}\big[B^{2(1+\alpha)}\big]$ is independent of $n$ and $\alpha > 1$.
\section{Proof of Proposition~\ref{ps6}}\label{ProofPS6}
  Denoting \small$X_{N,ij} = A_{ij}(N)v_j(N)$\normalsize, then \small$\{X_{N,ij}$: $1\leq j \leq N \}$\normalsize~is an independent zero-mean triangular array, for $i=1,2,\dots,n$. 
  We claim that \small$\{X_{N,ij}$: $1\leq j \leq N \}$\normalsize~satisfies two conditions in {Proposition}~\ref{ps2}, $\forall i$.
  First, we note that \small$\sum_{j=1}^{N}\E[X^2_{N,ij}] = \sum_{j=1}^{N}\E[A^2_{ij}(N)]v^2_j(N) = \frac{1}{n}\sum_{j=1}^{N}v^2_j(N) = \frac{1}{\rho}\lL\bv^2(N)\lR \rightarrow \frac{s^2_0}{\rho}$ as $n \rightarrow \infty$\normalsize.
  Second, applying Proposition~\ref{ps5} to \small$A_{ij}(N)$\normalsize~gives \small$\E[A^{2+2\alpha}_{ij}(N)] = o(n^{-2})$\normalsize, leading to   \small$\E[X^{2+2\alpha}_{N,ij}] = \E[A^{2+2\alpha}_{ij}(N)] |v_j(N)|^{2+2\alpha}  \leq o(n^{-2})\frac{n}{\rho}\frac{1}{N}\|\bv(N)\|^{2+2\alpha}_{2+2\alpha} = o(n^{-1})$\normalsize, 
 where the last equality holds because \small$\limsupN  \frac{1}{N}\|\bv(N)\|^{2+2\alpha}_{2+2\alpha} < \infty$\normalsize~and $\rho$ is a constant.
  Applying {Proposition}~\ref{ps2} to \small$\{X_{N,ij}: 1\leq j \leq N\}$\normalsize~leads to \small$[\bA(N)\bv(N)]_i = \sum_{j=1}^{N}X_{N,ij} \xRightarrow{\text{d}} \cN\left(0,\frac{s^2_0}{\rho}\right)$\normalsize~as \small$N \rightarrow \infty$\normalsize, $\forall i$.
  Hence, \small$\widehat{\bA(N)\bv(N)} \xRightarrow{d} \mathcal{N}\left(0,\frac{s^2_0}{\rho}\right)$\normalsize~as $N \rightarrow \infty$.
\section{Well-known Lemmas}\label{AppLM}
    \begin{lemma}(Stein's Lemma \cite{stein1972})
       \label{lm2}
        For jointly zero-mean Gaussian random variables $Z_1$ and $Z_2$, and any function $\phi:\R \rightarrow \R$, where $\E[\phi'(Z_2)]$ and $\E[Z_1\phi(Z_2)]$ exist, the following holds $\E[Z_1\phi(Z_2)] = \Cov(Z_1,Z_2) \E [\phi'(Z_2)]$, where $\Cov(Z_1,Z_2)$ is the covariance between $Z_1$ and $Z_2$. 
    \end{lemma}
    \begin{lemma}\label{lm3}(Strong Law of Large Number \cite{Hu1997}) 
    Let $\{X_{n,m}:1\leq m \leq n\}$ be a triangular array of random variables with $(X_{n,1},X_{n,2},\dots,X_{n,n})$ mutually independent with zero-mean for each $n$ and $\frac{1}{n}\sum_{m=1}^{n}\E[|X_{n,m}|^{2+\kappa}] \leq cn^{\kappa/2}$ for some $0 < \kappa <1$ and $c <\infty$. Then $\limn \frac{1}{n}\sum_{m=1}^{n}X_{n,m} \das 0$.
    \end{lemma}    
    \begin{lemma}
    \label{lm4} (Holder's inequality \cite{POZNYAK2009}) For random variables $X$ and $Y$, $\E[|X+Y|^r] \leq c_r(\E[|X|^r] + \E[|Y|^r]),$ where $c_r$ = 1 if $0 < r \leq 1$ and $c_r = 2^{r-1}$ otherwise. In particular, the inequality becomes  $\E[|X+y|^r] \leq c_r(\E[|X|^r] + |y|^r])$ when $Y = y$ being a constant. 
    \end{lemma}
    \begin{lemma} 
    \label{lm5}(Lyapunov's inequality \cite{POZNYAK2009})
    Suppose a random variable $X$ and a constant $\kappa$ with $0 < \kappa <1$, then  
    $|\E[X]|^{2+\kappa} \leq \E[|X|^{2+\kappa}].$
    \end{lemma}
\begin{proposition}\label{ps3}
Suppose the \sloppy $\bP_{\bM(n)}$ $=$ $\big(\frac{1}{\sqrt{n}} \bV(n)\big)\big(\frac{1}{\sqrt{n}}\bV(n)\big)^*$, where $\bM(n) \in \mathbb{R}^{n \times t}$ $(t \leq n)$, $t$ is a fixed constant, and $\bV(n)$$=$$[\bv_1(n), \bv_2(n), \dots, \bv_t(n)]$ $\in \mathbb{R}^{n \times t}$ is an orthogonal basis of $\bM(n)$ such that $\bV^*(n) \bV(n) = n \bI$. 
If we let $\ba(n) \in \mathbb{R}^{n \times 1}$ be a random vector with independent entries, which have zero mean and finite variance $\sigma_a^2$, then $\limn \bP_{\bM}(n)\ba(n) \das \mathbf{0}_n$, where $\mathbf{0}_n$ is the $n \times 1$ all-zero vector. 
\end{proposition}
\begin{proof}
	Denoting $\widetilde{\ba}(n) = \frac{\ba(n)}{\|\ba(n)\|_2}$ yields
    $\bP_{\bM(n)}\ba(n) = \bV(n)\frac{\|\ba(n)\|_2}{\sqrt{n}}\left(\frac{1}{\sqrt{n}}\bV^*(n)\right)\widetilde{\ba}(n)$. The proposition follows from the fact that $\frac{\|\ba(n)\|_2}{\sqrt{n}} \das \sigma_a$ and $\left(\frac{1}{\sqrt{n}}\bV^*(n)\right)\widetilde{\ba}(n) \das \mathbf{0}_n$ as $n \rightarrow \infty$. 
\end{proof}
\section{Proof of SE in \eqref{eq:SE} Theorem~\ref{theorem1}}
 \label{AppendixA} 
Substituting $\phi_b(\bu^t_i) = \left(b^t_i\right)^2$ into \eqref{b:b} gives $\lim_{n \rightarrow \infty} \langle \bb^t, \bb^t\rangle \stackrel{a.s}{=} \E\left[ \sigma^2_t\widetilde{Z}^2_t\right]= \sigma^2_t$. 
Using \eqref{eq:b2} with $t_1 = t_2 = t$ yields $\lim_{N \rightarrow \infty} \frac{1}{\rho} \langle \bq^t, \bq^t \rangle = \sigma^2_t$.
Then substituting $\phi_h(\bv^{t-1}_i) = f_t^2(h^t_i,x_{0i}) = (q^t_i)^2$ into \eqref{b:h} leads to $\lim_{N \rightarrow \infty} \langle \bq^t, \bq^t \rangle \stackrel{a.s}{=} \E\left[f_t^2(\tau_{t-1}Z_{t-1},X_0)\right]$, resulting in $\sigma^2_t = \frac{1}{\rho}\E\left[f_t^2(\tau_{t-1}Z_{t-1},X_0)\right]$. 
Showing the rest half $\tau^2_t = \E \left[g_t^2(\sigma_tZ,W)\right]$ of the SE in \eqref{eq:SE} follows from the exactly same procedure as the above.
Setting $\phi_h(\bv^t_i) = (h^{t+1}_i)^2$ in \eqref{b:h} gives $\lim_{N \rightarrow \infty} \langle \bh^{t+1},\bh^{t+1}\rangle = \tau^2_t$. Using \eqref{eq:h2} with $t_1=t_2=t$ yields $\lim_{N \rightarrow \infty} \langle \bm^t, \bm^t \rangle = \tau^2_t$. 
Applying \eqref{b:b} to $\phi_b(\bu^t_i) = g_t^2(b^t_i,w_i)$ yields $\lim_{N \rightarrow \infty} \langle \bm^t, \bm^t \rangle = \E\left[ g^2_t(\sigma_tZ,W)\right]$. Therefore, $\tau^2_t = \E\left[ g^2_t(\sigma_tZ,W)\right]$, concluding the proof.       
\section{Proof of Theorem \ref{theorem1}: Steps 3 and 4}\label{AppendixC}
     \subsection{Step 3: We show a), b), c), and d) of Theorem~\ref{theorem1} conditioning on $\cF_{t,t} = \{\bb^0,\dots,\bb^{t-1},\bm^0,\dots,\bm^{t-1},\bh^1,\dots,\bh^t,\bq^0,\dots,\bq^t,\bx_0,\bw\}$.}
    a)
    Note that conditioning on $\cF_{t,t}$ is equivalent to conditioning on $\cP_{\bX_t,\bY_t} = \{\bA|\bA^*\bM_t = \bX_t, \bA\bQ_t = \bY_t\}$. 
    Applying Proposition~\ref{ps4} to obtain the conditional distribution 
    $\bA|_{\cF_{t,t}}$ and following the same procedure as in \cite[Lemma 1a]{Bayati11}, the conditional distribution of $\bb^t$ on $\cF_{t,t}$ is expressed as 
    \begin{equation}
    \label{eq:3.1}
        \bb^t|_{\cF_{t,t}} \ddis  \sum_{j=0}^{t-1}\beta_j\mathbf{b}^j + \widetilde{\bA}\bq^{t}_{\perp} - \bP_{\bM_t}\widetilde{\bA}\bq^t_{\perp}+ \mathbf{M}_t\overrightarrow{\bo_t}(1).
    \end{equation}
    By {Proposition}~\ref{ps3} in Appendix~\ref{AppLM}, the third term on the {r.h.s} of \eqref{eq:3.1} converges to $\limn \bP_{\bM_t}\widetilde{\bA}\bq^t_{\perp} \das \mathbf{0}_n$.
    Similar to {Step}~2a), we verify the convergence $\limn \bM_t\overrightarrow{\bo}_t(1)\das \mathbf{0}_n$ by characterizing the expectation and variance of its empirical distribution $\widehat{\bM_t\overrightarrow{\bo}_t(1)}$ as $n \rightarrow \infty$.
    Indeed, $\limn |\lL \bM_t\overrightarrow{\bo}_t(1) \lR| \leq \limn |o(1)|\frac{1}{n}\sum_{i=1}^{n}\left|\sum_{j=0}^{t-1}m^j_i\right|$ $\leq$ $\limn |o(1)|\sum_{j=0}^{t-1}\frac{1}{n}\sum_{i=1}^{n}|m^j_i| \das  0,$ 
    where the last equality holds because applying $\phi_b (b^j_i,w_i) = g_j(b^j_i,w_i)$ to the induction hypothesis of \eqref{b:b}, for $j < t$, leads to $\limn \frac{1}{n}\sum_{i=1}^{n}|m^j_i| \das \E[|g_j(\sigma_j\widetilde{Z}_j,W)|] < \infty$.  
    Hence, $\limn \lL \bM_t\overrightarrow{\bo}_t(1) \lR \das 0$. For the variance of $\widehat{\bM_t\overrightarrow{\bo}_t(1)}$, we get  
    \begin{subequations}
    \label{eq:3.3}
    \beq
    \d4\d4\limn \lL \bM_t \overrightarrow{\bo}_t(1) \lR_2 \d4&=&\d4 \limn \frac{1}{n}[o(1)]^2 \sum_{i=1}^{n}\left( \sum_{j=0}^{t-1}m^j_i \right)^2,\nonumber\\ 
    \d4&\leq&\d4 \limn [o(1)]^2 t \frac{1}{n}\sum_{i=1}^{n} \sum_{j=0}^{t-1}(m^j_i)^2, \label{eq:3.3a}\\
    \d4& = &\d4 \limn [o(1)]^2t\sum_{j=0}^{t-1}\lL\bm^j,\bm^j\lR \das 0, \label{eq:3.3b}
    \eeq
    \end{subequations}
    where \eqref{eq:3.3a} follows from the Cauchy-Schwarz inequality and \eqref{eq:3.3b} holds because $\limn \lL\bm^j,\bm^j \lR \das \E[g^2_j(\sigma_j\widetilde{Z}_j,W)] < \infty$, $\forall j$, which is a direct consequence of the induction hypothesis of \eqref{b:b} with $\phi_b(b^j_i,w_i)=g^2_j(b^j_i,w_i)$.
    Hence, \eqref{eq:3.3b} is equivalent to $\limn \lL \bM_t \overrightarrow{\bo}_t(1) \lR_2 \das 0$. 
    Therefore, $\limn \bM_t \overrightarrow{\bo}_t(1) \das \mathbf{0}_n$, implying 
    \begin{equation}
        \label{eq:3.4}
        \bb^t|_{\cF_{t,t}} \xRightarrow{d} \sum_{j=0}^{t-1}\beta_j\bb^j + \widetilde{\bA}\bq^t_{\perp}.
    \end{equation}
    b) Note that by the induction hypothesis of \eqref{b:h} for  $\phi_h(h^t_i,x_{0i}) = {|f_t(h^t_i,x_{0i})|^{2+2\alpha}}$, we get $\limN \frac{1}{N} \sumN |q^t_i|^{2+2\alpha} \das \E[{|f_t(\tau_{t-1}Z_{t-1},X_0)|^{2+2\alpha}}] < \infty.$ 
    On the other hand, $\sumN |q^t_{\perp i}|^{2+2\alpha} < \sumN |q^t_{i}|^{2+2\alpha}$.
    Thus, we have     
    $\limsupN \frac{1}{N} \sumN |q^t_{\perp i}|^{2+2\alpha} < \infty$, which concludes \eqref{eq:alphaQ}.
    c) For $t_1 < t$ and $t_2=t$, we obtain
    \small\begin{subequations}
    \beq
     \d4\limn \lL\bb^{t_1},\bb^t\lR \d4&\ddis&\d4\d4 \limn \sum_{j=0}^{t-1} \beta_j\lL\bb^{t_1},\bb^j\lR \!+\!\limn  \lL \bb^{t_1},\widetilde{\bA}\bq^t_{\perp}\lR, \label{eq:3.5a}\\
    \d4&\das&\d4  \sum_{j=0}^{t-1} \beta_j\!\limN\!\!\frac{\lL\bq^{t_1},\bq^j\lR}{\rho} \!+\! \limn \frac{\bb^{t_1^*}\widetilde{\bA}\bq^t_{\perp}}{n} \label{eq:3.5b}, 
    \eeq
    \end{subequations}\normalsize 
    where \eqref{eq:3.5a} follows from \eqref{eq:3.4} and \eqref{eq:3.5b} results from the induction hypothesis \eqref{eq:b2} for $t_1 < t$ and $t_2 = j < t$. 
    Now, using {Proposition}~\ref{ps7}, we get $ \frac{\bb^{t_1^*}}{\|\bb^{t_1}\|_2}\widetilde{\bA}\frac{\bq^t_{\perp}}{\|\bq^t_{\perp}\|_2} \ddis \frac{Z}{\sqrt{n}}$, where $Z \sim \cN(0,1)$. Hence, $\frac{\bb^{t_1^*}\widetilde{\bA}\bq^t_{\perp}}{n} \ddis  \frac{\|\bb^{t_1}\|_2}{\sqrt{n}}\frac{\|\bq^t_{\perp}\|_2}{\sqrt{N}}\frac{1}{\sqrt{\rho}}\frac{Z}{\sqrt{n}}$, i.e., 
    \small\begin{equation}
    \frac{\bb^{t_1^*}\widetilde{\bA}\bq^t_{\perp}}{n} \ddis \frac{1}{\sqrt{\rho}} \sqrt{\lL\bb^{t_1},\bb^{t_1} \lR \lL\bq^t_{\perp},\bq^t_{\perp}\lR} \frac{Z}{\sqrt{n}}.\label{eq:3.6}
    \end{equation}\normalsize
    By the induction hypothesis of \eqref{eq:b2}, we have $\limn \lL\bb^{t_1},\bb^{t_1} \lR = \limn \frac{\lL \bq^{t_1},\bq^{t_1} \lR}{\rho} < \infty$. 
    Moreover, the $\lL\bq^t_{\perp},\bq^t_{\perp}\lR$ converges to $\limN \lL\bq^t_{\perp},\bq^t_{\perp}\lR < \limN \lL\bq^t,\bq^t\lR < \infty$ because using the induction hypothesis \eqref{b:h} we have $\lL\bq^t,\bq^t\lR  = \frac{1}{N}\sumN f^2_t(h^t_i,x_{0i}) \das \E[f^2_t(\tau_{t-1}Z_{t-1},X_0)] < \infty$.
    Thus, for $t_1 <t$, 
    \begin{equation}
        \label{eq:3.7}
      \limn  \lL\widetilde{\bA}\bq^t_{\perp},\bb^{t_1}\lR \das 0. 
    \end{equation} Substituting \eqref{eq:3.7} into \eqref{eq:3.5b} gives  $\limn \lL\bb^{t_1},\bb^t\lR \das  \sum_{j=0}^{t-1}\beta_j\limn\frac{\lL\bq^{t_1},\bq^j\lR}{\rho} = \limN \frac{\lL \bq^{t_1}, \bq^t_{||}\lR}{\rho}$ due to \eqref{eq:projection1}, implying 
    \begin{subequations}
    \label{eq:3.8}
    \beq 
    \d4\d4\limn \lL\bb^{t_1},\bb^t\lR  \d4&\das&\d4 \limN \frac{\lL \bq^{t_1}, \bq^t_{||} \lR}{\rho} + \limN \frac{\lL \bq^{t_1}, \bq^t_{\perp} \lR}{\rho} ,\label{eq:3.8b}\\
    \d4&=&\d4 \limN \frac{\lL \bq^{t_1}, \bq^t\lR}{\rho}, \label{eq:3.8c}
    \eeq 
    \end{subequations}
    where \eqref{eq:3.8b} follows from the fact that $\bq^j$ is orthogonal to $\bq^t_{\perp}$, for $j < t$, and \eqref{eq:3.8c} holds due to \eqref{eq:projection2}, 
    concluding \eqref{eq:b2} when $t_1 < t$ and $t_2 = t$. 
    For the case of $t_1 = t_2 =t$, it is similarly given by 
    $\limn \lL \bb^t,\bb^t \lR$ $\ddis$ $ \sum_{ i,j = 0}^{t-1}\beta_i\beta_j \limn  \lL\bb^i,\bb^j \lR +$ $2\sum_{i=0}^{t-1}\beta_i \limn\lL\bb^i, \widetilde{\bA}\bq^t_{\perp}\lR +$ $\limn \lL\widetilde{\bA}\bq^t_{\perp},\widetilde{\bA}\bq^t_{\perp}\lR$ due to \eqref{eq:3.4}. 
    Then, by \eqref{eq:3.7}, the following holds 
    \small\begin{equation}
    \d4 \limn \lL \bb^t,\bb^t \lR \das \sum_{ i,j = 0}^{t-1}\!\beta_i\beta_j \!\limn  \lL\bb^i,\bb^j \lR \!+\! \limn \lL\widetilde{\bA}\bq^t_{\perp},\widetilde{\bA}\bq^t_{\perp}\lR. \label{eq:3.9b}
    \end{equation}\normalsize
   Using {Proposition}~\ref{ps6}, $\widehat{\widetilde{\bA}\bq^t_{\perp}} \xRightarrow{d} \cN\left(0,\limN \frac{\lL\bq^t_{\perp},\bq^t_{\perp}\lR}{\rho}\right)$. 
    Thus, the second moment of $\widehat{\widetilde{\bA}\bq^t_{\perp}}$ is  
    \begin{equation}
    \label{eq:3.9.1}
    \limn \lL \widetilde{\bA}\bq^t_{\perp}, \widetilde{\bA}\bq^t_{\perp}\lR \das  \limN \frac{\lL\bq^t_{\perp},\bq^t_{\perp}\lR}{\rho}.
    \end{equation} 
    Now, incorporating \eqref{eq:3.9.1} in \eqref{eq:3.9b}, $\limn \lL \bb^t,\bb^t \lR \das \sum_{ i,j = 0}^{t-1}\beta_i\beta_j \limn  \lL\bb^i,\bb^j \lR+\limN\frac{\lL\bq^t_{\perp},\bq^t_{\perp}\lR}{\rho}$, resulting in   
    \small\begin{subequations}
    \label{eq:3.10}
    \beq 
    \d4\d4\limn \lL \bb^t,\bb^t \lR\d4&\das&\d4 \!\!\! \sum_{ i,j = 0}^{t-1}\beta_i\beta_j \limN \! \frac{\lL\bq^i,\bq^j \lR}{\rho} \!+\! \limN\!\frac{\lL\bq^t_{\perp},\bq^t_{\perp}\lR}{\rho}, \label{eq:3.10b}\\
    \d4&\das&\d4 \!\!\!\limN \frac{\lL\bq^t_{||},\bq^t_{||}\lR}{\rho} + \limN\frac{\lL\bq^t_{\perp},\bq^t_{\perp}\lR}{\rho}, \nonumber\\ 
    &\das&\d4\!\!\!\limN\frac{\lL\bq^t,\bq^t\lR}{\rho}, \nonumber
    \eeq 
    \end{subequations}\normalsize   
    where \eqref{eq:3.10b} is due to the induction hypothesis \eqref{eq:b2} for $0 \leq t_1 = i, t_2 = j \leq t-1$. 
   This completes the proof of \eqref{eq:b2} at the $t$th iteration.   \\

    d) Defining $\limN \frac{\lL \bq^t_{\perp}, \bq^t_{\perp}\lR}{\rho} \das \gamma^2_t$, we can write by \eqref{eq:3.9.1} that $\widehat{\widetilde{\bA}\bq^t_{\perp}} \xRightarrow{d} \cN(0,\gamma^2_t)$.
    Using \eqref{eq:3.4} in conjunction with the latter, we get 
    \begin{equation}
        \label{eq:3.12}
        b^t_i|_{\cF_{t,t}} \xRightarrow{d} \sum_{j=0}^{t-1}\beta_jb^j_i + \gamma_tZ, \text{ for } i= 1,2,\dots,n,
    \end{equation}
     where $Z \sim \cN(0,1)$. 
     Similar to Step 1d), using \eqref{eq:3.12} $\bu^t_i \xRightarrow{d} \widetilde{\bu}^t_i$, where  
      $\bu^t_i = (b^0_i \dots, b^{t}_i,w_i)$ and $\widetilde{\bu}^t_i = (b^0_i,\dots,b^{t-1}_i,\sum_{j=0}^{t-1}\beta_j b^j_i + \gamma_tZ,w_i)$, $\forall i$.
     To prove \eqref{b:b}, we first claim that $\limn \frac{1}{n} \sumn \phi_b(\widetilde{\bu}^t_i) - \mathbb{E}\!\left[\phi_b(\sigma_0\widetilde{Z}_0,\dots, \sigma_t\widetilde{Z}_t, W)\right] \das 0$. By the triangular inequality, \small$ \bigg| \limn \frac{1}{n} \sumn \phi_b(\widetilde{\bu}^t_i) - \mathbb{E}\!\left[\phi_b(\sigma_0\widetilde{Z}_0,\dots, \sigma_t\widetilde{Z}_t, W)\right] \bigg| \leq X^t_1 + X^t_2$\normalsize, 
     where \small$X^t_1$ $=$ $\bigg| \frac{1}{n}\sumn \left(\phi_b(\widetilde{\bu}^t_i) - \widetilde{\phi}_b(\bu^{t-1}_i)\right)\bigg|$\normalsize, 
    \small$X^t_2$ $=$ $\bigg| \frac{1}{n}\sumn  \widetilde{\phi}_b(\bu^{t-1}_i) - \E[\phi_b(\sigma_0\widetilde{Z}_0,\dots,\sigma_t\widetilde{Z}_t,W)]\bigg|$\normalsize, and $\widetilde{\phi}_b(\bu^{t-1}_i) = \E_{Z}[\phi_b(\widetilde{\bu}^t_i)]$. 
    Similar to {Step}~1d), we verify $\limn X^t_1 \das 0$ and $\limn X^t_2 \das0$. 
   First showing $\limn X^t_1 \das 0$ is of interest.
{   By \eqref{eq:controlled2}, 
   $|\phi_b(\widetilde{\bu}^t_i)| 
   \leq c^t_{1} \exp \Big( c^t_2\Big( \sum_{j=0}^{t-1}|b^t_i|^{\lambda} + \Big|\sum_{j=0}^{t-1} \beta_j b^j_i +\gamma_t Z \Big|^{\lambda}+ |w_i|^{\lambda} \Big) \Big),$  where $c^t_1>0$, $c^t_2>0$, and $1 \leq \lambda <2$ are constants. 
   Using the inequality $\|\bx\|^{\lambda}_1 \leq (t+1)^{\lambda-1}\|\bx\|^{\lambda}_{\lambda}$ for $\bx \in \R^{(t+1) \times 1}$, we get 
   $|\phi_b(\widetilde{\bu}^t_i)| \leq c^t_{1} \exp \Big( c^t_2  \Big( \sum_{j=0}^{t-1}(1+(t+1)^{\lambda-1}|\beta_j|^{\lambda})|b^t_i|^{\lambda} + (t+1)^{\lambda-1} |\gamma_t|^{\lambda} |Z|^{\lambda} + |w_i|^{\lambda} \Big) \Big).$
   Hence, $\E_{Z}[|\phi_b(\widetilde{\bu}^t_i)|^{2+\kappa}] \leq c^t_{3} \exp\left(c^t_4\left(\sum_{j=0}^{t-1}|b^j_i|^{\lambda} + |w_i|^{\lambda}\right)\right)\E_Z\left[\exp(c^t_4|Z|^{\lambda})\right]$, where $0< \kappa <1$, $c^t_3 = (c^t_1)^{2+\kappa}$, and $c^t_4 = (2+\kappa) c^t_2 \max \Big\{ 1+(t+1)^{\lambda-1}|\beta_0|^{\lambda},\dots,1+(t+1)^{\lambda-1}|\beta_{t-1}|^{\lambda},(t+1)^{\lambda-1}|\gamma_t|^{\lambda} \Big\}$,  resulting in} 
\begin{equation}
    \E_{Z}[|\phi_b(\widetilde{\bu}^t_i)|^{2+\kappa}] \leq c^t_5\exp\Big(c^t_4\Big(\sum_{j=0}^{t-1}|\beta_j b^j_i|^{\lambda} + |w_i|^{\lambda}\Big)\Big), \label{eq:3.14a}
\end{equation} 
and $c^t_5 = c^t_3\E_Z\left[\exp(c^t_4\gamma^{\lambda}_t|Z|^{\lambda})\right]$ is constant. 
We define $X^t_{n,i} = \phi_b(\widetilde{\bu}^t_i) - \widetilde{\phi}_b(\bu^{t-1}_i) = \phi_b(\widetilde{\bu}^t_i) - \E_{Z}[\phi_b(\widetilde{\bu}^t_i)]$ such that $X_1^t = |\frac{1}{n}\sumn X^t_{n,i}|$. 
To prove $\limn X^t_1 \das 0$, we show that $\{X^t_{n,i}\}_{i=1}^{n}$ satisfy {Lemma}~\ref{lm3} in Appendix~\ref{AppLM}. 
Indeed, $\E_Z[|X^t_{n,i}|^{2+\kappa}]$ is upper bouned as follows,  
    \small\begin{subequations}
    \label{eq:3.15}
    \beq 
    \d4\d4\d4\E_Z[|X^t_{n,i}|^{2+\kappa}] 
    \d4\!&\leq&\d4\!  2^{1+\kappa}\!\left(\E_Z\left[|\phi_b(\widetilde{\bu}^t_i)|^{2+\kappa}\right] \!+\! \left|\E_Z[\phi_b(\widetilde{\bu}^t_i)]\right|^{2+\kappa} \!\!\right), \label{eq:3.15b}\\  \d4\!&\leq&\d4\!2^{2+\kappa}\E_Z\left[|\phi_b(\widetilde{\bu}^t_i)|^{2+\kappa}\right], \label{eq:3.15c}\\ 
    \d4\!&\leq&\d4\! c^t_6\exp\Big(c^t_4\Big(\sum_{j=0}^{t-1}|\beta_j b^j_i|^{\lambda} + |w_i|^{\lambda}\Big)\Big),\label{eq:3.15e}
    \eeq 
    \end{subequations}\normalsize 
    where  \eqref{eq:3.15b} follows from {Lemma}~\ref{lm4} (Holder's inequality) in Appendix~\ref{AppLM}, \eqref{eq:3.15c} follows from {Lemma}~\ref{lm5} (Lyapunov's inequality) in Appendix~\ref{AppLM}, and \eqref{eq:3.15e} follows from \eqref{eq:3.14a} with $c^t_6 = 2^{2+\kappa}c^t_5$. 
    We denote the last term of \eqref{eq:3.15e} as \small$\psi_b(\bu^{t-1}_i) =  c^t_6\exp\left(c^t_4\left(\sum_{j=0}^{t-1}|\beta_j b^j_i|^{\lambda} + |w_i|^{\lambda}\right)\right)$\normalsize. Then, $\psi_b(\bu^{t-1}_i)$ is a controlled function. 
    From \eqref{eq:3.15e}, we get, for $n$ is sufficiently large,  
    \small\begin{subequations}
    \label{eq:3.16}
    \beq 
    \d4\d4\frac{1}{n}\sumn \E_{Z}[|X^t_{n,i}|^{2+\kappa}] \d4&\leq&\d4 \frac{1}{n}\sumn \psi_b(\bu^{t-1}_i),\nonumber\\
    \d4&\das&\d4 \E[\psi_b(\sigma_0\widetilde{Z}_0,\dots,\sigma_{t-1}\widetilde{Z}_{t-1},W)],\label{eq:3.16b}\\
    \d4&<&\d4cn^{\kappa/2},\label{eq:3.16c}
    \eeq 
    \end{subequations}\normalsize
     where $c$ is a positive constant,
     \eqref{eq:3.16b} is due to the induction hypothesis \eqref{b:b},  
     and \eqref{eq:3.16c} holds because $\E[\psi_b(\sigma_0\widetilde{Z}_0,\dots,\sigma_{t-1}\widetilde{Z}_{t-1},W)] = c^t_7 < \infty$ and there exists $n_t$, a positive constant, such that $ c^t_7 < cn^{\kappa/2}$ for $n > n_t$. 
    By {Lemma}~\ref{lm3} in Appendix~\ref{AppLM}, we get $\limn \frac{1}{n}\sumn X^t_{n,i} \das 0$, implying 
    \begin{equation}
        \label{eq:3.17}
        \frac{1}{n}\sumn \left(\phi_b(\widetilde{\bu}^t_i) - \widetilde{\phi}_b(\bu^{t-1}_i)\right) \das 0, 
    \end{equation}
    which proving $\limn X^t_1 \das 0$.

     Now, showing $\limn X^t_2 \das 0$ is of interest. 
     By the induction hypothesis \eqref{b:b}, $\limn \frac{1}{n} \sumn \widetilde{\phi}_b(\bu^{t-1}_i) \das \E[\widetilde{\phi}_b(\sigma_0\widetilde{Z}_0,\dots,\sigma_{t-1}\widetilde{Z}_{t-1},W)]$, resulting in  
    \small\begin{multline}\label{eq:3.18c}
    \d4\limn \frac{1}{n} \sumn \widetilde{\phi}_b(\bu^{t-1}_i) \\  
    \begin{aligned}
    &=\E\left[ \E_Z\left[\phi_b(\sigma_0\widetilde{Z}_0,\dots,\sigma_{t-1}\widetilde{Z}_{t-1},\sum_{j=0}^{t-1}\beta_j\sigma_j\widetilde{Z}_j + \gamma_tZ,W)\right]\right],\\
    &=\E\left[\phi_b(\sigma_0\widetilde{Z}_0,\dots,\sigma_{t-1}\widetilde{Z}_{t-1},\sum_{j=0}^{t-1}\beta_j\sigma_j\widetilde{Z}_j + \gamma_tZ,W)\right],
    \end{aligned}
    \end{multline}\normalsize
    where \eqref{eq:3.18c} follows from the substitution $\widetilde{\phi}_b(\bu^{t-1}_i) = \E_{Z}[\phi_b(\widetilde{\bu}^t_i)]$. 
    Therefore, showing $\limn X^t_2 \das 0$ is equivalent to proving $\sum_{j=0}^{t-1}\beta_j\sigma_j\widetilde{Z}_j + \gamma_tZ = \sigma_t\widetilde{Z}_t$, where $\widetilde{Z}_t \sim \cN(0,1)$ and $\sigma_t$ is defined in \eqref{eq:SE}. 
    In particular, for $\phi_b(\bu^t_i) = (b^t_i)^2$, we get $\phi_b(\widetilde{\bu}^t_i) = \left(\sum_{j=0}^{t-1}\beta_jb^j_i + \gamma_tZ \right)^2$ because $\widetilde{\bu}^t_i = (b^0_i,\dots,b^{t-1}_i,\sum_{j=0}^{t-1}\beta_j b^j_i + \gamma_tZ,w_i)$. 
    Combining \eqref{eq:3.17} and \eqref{eq:3.18c}, 
    \small\begin{multline}
        \label{eq:3.19}
        \limn \lL\bb^t,\bb^t \lR = \limn \frac{1}{n}\sumn \phi_b(\bu^t_i) \\ {\ddis} \limn \frac{1}{n}\sumn \phi_b(\widetilde{\bu}^t_i)
        \das   \E\Big[\Big(\sum_{j=0}^{t-1}\beta_j\sigma_j\widetilde{Z}_j + \gamma_tZ\Big)^2\Big]. 
    \end{multline}\normalsize
    Using \eqref{eq:b2}, $\limn \lL\bb^t,\bb^t \lR \das \limN \frac{\lL\bq^t,\bq^t \lR}{\rho} = \sigma_t^2$, where the last equality holds because by the induction hypothesis \eqref{b:h} for $\phi_h(\bv^{t-1}_i) = f^2_t(h^{t}_i,x_{0i})$ in \eqref{b:h}, $\frac{1}{\rho}\limN \lL \bq^t,\bq^t \lR \das \frac{1}{\rho}\E[f^2_t(\tau_{t-1}Z,X_0)] = \sigma^2_t$. 
    Hence, $\E\left[\left(\sum_{j=0}^{t-1}\beta_j\sigma_j\widetilde{Z}_j + \gamma_tZ\right)^2\right] \das \sigma^2_t$, implying $\sum_{j=0}^{t-1}\beta_j\sigma_j\widetilde{Z}_j + \gamma_tZ = \sigma_t\widetilde{Z}_t$ due to \eqref{eq:3.19}, verifying that $\limn X^t_2 \das 0$, which completes the proof of \eqref{b:b}. 
      \subsection{Step 4: We show a), b), c), and d) of Theorem~\ref{theorem1} conditioning on $\cF_{t+1,t} = \{\bb^0,\dots,\bb^{t},\bm^0,\dots,\bm^{t},\bh^1,\dots,\bh^t,\bq^0,\dots,\bq^t,\bx_0,\bw\}$.} 
    The proof of Step 4 is similar to the proof of Step 3. Thus, we only present the features that are unique in Step 4. 
    a) Similar to {Step} 3a), using Proposition~\ref{ps4} to characterize $\bA|_{\cF_{t+1,t}}$ and following the same procedure as in \cite[Lemma 1a]{Bayati11}, the $\bh^{t+1}|_{\cF_{t+1,t}}$ is  
    \small\begin{equation}
        \label{eq:4.1}
        \bh^{t+1}|_{\cF_{t+1,t}} \ddis \sum_{j=0}^{t-1}\zeta_j\bh^{j+1} +\widetilde{\bA}^*\bm^t_{\perp} - \bP_{\bQ_{t+1}}\widetilde{\bA}^*\bm^t_{\perp} + \bQ_{t}\overrightarrow{\bo}_t(1).
    \end{equation}\normalsize
    By {Proposition}~\ref{ps3} in Appendix~\ref{AppLM}, $\limN \bP_{\bQ_{t+1}}\widetilde{\bA}^*\bm^t_{\perp} \das \mathbf{0}_N$. 
    Similar to Step 3a), we verify that $\limN \bQ_{t}\overrightarrow{\bo}_{t}(1) \das \mathbf{0}_N$ by  characterizing (i) the expectation of the empirical distribution $\widehat{\bQ_{t}\overrightarrow{\bo}_{t}(1)}$ is bounded as 
    \small$\limN |\lL \bQ_{t}\overrightarrow{\bo}_t(1) \lR| \leq \limN |o(1)| \sum_{j=0}^{t-1} \frac{1}{N}\sumN |q^j_i|$ $\das 0$\normalsize~and (ii) the empirical variance of $\widehat{\bQ_{t}\overrightarrow{\bo}_{t}(1)}$ is bounded an converges to  \small$\limN \lL \bQ_{t}\overrightarrow{\bo}_{t}(1)\lR_2 \leq \limN [o(1)]^2t\sum_{j=0}^{t-1} \lL\bq^j,\bq^j \lR \das 0$\normalsize.
    Therefore, using $\limN \bQ_{t}\overrightarrow{\bo}_{t}(1) \das \mathbf{0}_N$, we get 
    \begin{equation}
        \label{eq:4.4}
                \bh^{t+1}|_{\cF_{t+1,t}} \xRightarrow{d} \sum_{j=0}^{t-1}\zeta_j\bh^{j+1} +\widetilde{\bA}^*\bm^t_{\perp},
    \end{equation} which completes the proof of \eqref{eq:h1}.\\
    b) Using \eqref{b:b} for $\phi_b(b^t_i,w_i) = |g_t(b^t_i,w_i)|^{2+2\alpha}$, we get $\limn \frac{1}{n} \sumn |m^t_i|^{2+2\alpha} \das \E[|g_t(\sigma_t\widetilde{Z}_t,W)|^{2+2\alpha}] < \infty$. 
    Because $\sumn |m^t_{\perp i}|^{2+2\alpha} < \sumn |m^t_{i}|^{2+2\alpha}$, the following holds $\limsupn \sumn |m^t_{\perp i}|^{2+2\alpha} < \infty$, which concludes \eqref{eq:alphaM}.
    c) For $t_1 <t$ and $t_2 =t$, we have \small$\limN \lL \bh^{t_1+1},\bh^{t+1} \lR \ddis \limN \sum_{j=0}^{t-1}\zeta_j\lL\bh^{t_1+1},\bh^j\lR + \limN \lL\bh^{t_1+1},\widetilde{\bA}^*\bm^t_{\perp}\lR$\normalsize~due to \eqref{eq:4.4}, resulting in   
    \small\begin{equation}
        \d4\limN \lL \bh^{t_1+1},\bh^{t+1} \lR\das \!\sum_{j=0}^{t-1}\zeta_j\limn \lL\bm^{t_1},\bm^j\lR \!+\! \limN  \frac{\bm^{t^*}_{\perp}\widetilde{\bA}\bh^{t_1+1}}{N}, \label{eq:4.5b}
    \end{equation}\normalsize
    where \eqref{eq:4.5b} is by the induction hypothesis \eqref{eq:h2}. 
    Note that $\frac{\bm^{t^*}_{\perp}}{\|\bm^t_{\perp}\|_2} \widetilde{\bA} \frac{\bh^{t_1+1}}{\|\bh^{t_1+1}\|_2} \ddis \frac{Z}{\sqrt{n}}$ due to Proposition~\ref{ps7}.
    The second term in \eqref{eq:4.5b} is represented as  
    \small\begin{multline}
    \limN \frac{\bm^{t^*}_{\perp}\widetilde{\bA}\bh^{t_1+1}}{N}  \ddis\limN  \frac{\|\bm^{t}_{\perp}\|_2}{\sqrt{n}}\frac{\|\bh^{t_1+1}\|_2}{\sqrt{N}} \frac{\sqrt{n}}{\sqrt{N}}\frac{Z}{\sqrt{n}},\\
    =\sqrt{\rho}\limN \sqrt{\lL\bm^{t}_{\perp},\bm^{t}_{\perp} \lR \lL \bh^{t_1+1},\bh^{t_1+1}\lR}\frac{Z}{\sqrt{n}} \das 0,  \label{eq:4.6d}   
    \end{multline}\normalsize
    Substituting \eqref{eq:4.6d} into \eqref{eq:4.5b} yields
    \small\begin{multline}
     \d4\d4\limN \lL \bh^{t_1+1},\bh^{t+1} \lR  =\limn \lL\bm^{t_1},\sum_{j=0}^{t-1}\zeta_j\bm^j \lR = \limn \lL \bm^{t_1}, \bm^t_{||}\lR,\\
     = \limn \lL \bm^{t_1}, \bm^t_{||}\lR + \limn \lL \bm^{t_1}, \bm^t_{\perp}\lR = \limn \lL \bm^{t_1}, \bm^t\lR, \nonumber   
    \end{multline}\normalsize
    concluding \eqref{eq:h2} when $t_1 < t$ and $t_2 =t$.  
    
    For $t_1 =t_2 = t$, by \eqref{eq:4.4},
    \small\begin{multline}
        \limN \lL \bh^{t+1},\bh^{t+1} \lR  \ddis \sum_{i,j = 0}^{t-1}\zeta_i\zeta_j \limN \lL \bh^{i+1},\bh^{j+1}\lR \\ + 2\sum_{i=0}^{t-1}\zeta_i\limN \lL\bh^{i+1},\widetilde{\bA}^*\bm^t_{\perp}\lR + \limN \lL\widetilde{\bA}^*\bm^t_{\perp},\widetilde{\bA}^*\bm^t_{\perp}\lR.
    \end{multline}\normalsize~
    Then, by \eqref{eq:4.6d}, the following holds 
    \small\begin{multline}\label{eq:4.9}
        \limN \lL \bh^{t+1},\bh^{t+1}\lR  \das \sum_{i,j = 0}^{t-1}\zeta_i\zeta_j  \limN \lL \bh^{i+1},\bh^{j+1}\lR + \\ \limN \lL\widetilde{\bA}^*\bm^t_{\perp},\widetilde{\bA}^*\bm^t_{\perp}\lR. 
    \end{multline}\normalsize 
    By {Proposition}~\ref{ps6}, the empirical distribution of $\widetilde{\bA}^*\bm^t_{\perp}$ converges to         $\widehat{\widetilde{\bA}^*\bm^t_{\perp}} \xRightarrow{d} \cN(0,\limn \lL \bm^t_{\perp},\bm^t_{\perp}\lR)$.
    Hence, the second moment of $\widehat{\widetilde{\bA}^*\bm^t_{\perp}}$ converges to 
    \begin{equation}
    \label{eq:4.10}
    \limN \lL\widetilde{\bA}^*\bm^t_{\perp},\widetilde{\bA}^*\bm^t_{\perp} \lR \das \limn \lL\bm^t_{\perp},\bm^t_{\perp}\lR. 
    \end{equation} 
    Substituting \eqref{eq:4.10} into \eqref{eq:4.9} leads to $\limN \lL \bh^{t+1},\bh^{t+1} \lR \das\sum_{i,j = 0}^{t-1}\zeta_i\zeta_j \limN \lL \bh^{i+1},\bh^{j+1}\lR+\limn \lL\bm^t_{\perp},\bm^t_{\perp}\lR$, implying
     \small\begin{multline}
      \d4\d4\limN \lL \bh^{t+1},\bh^{t+1} \lR \das  \sum_{i,j = 0}^{t-1}\zeta_i\zeta_j \limn \lL \bm^{i},\bm^{j}\lR + \limn \lL\bm^t_{\perp},\bm^t_{\perp}\lR, \\
     =\limn \lL\bm^t_{||},\bm^t_{||}\lR+\limn \lL\bm^t_{\perp},\bm^t_{\perp}\lR = \limn \lL \bm^t,\bm^t\lR. \nonumber 
     \end{multline}\normalsize
    Therefore, \eqref{eq:h2} also holds for $t_1 = t_2 =t$, which completes the proof.

    d) Defining $\limn \lL \bm^t_{\perp},\bm^t_{\perp} \lR \das \Gamma^2_t$, we can write \begin{equation}
    \label{eq:4.12}
            \widehat{\widetilde{\bA}^*\bm^t_{\perp}} \xRightarrow{d} \cN(0,\Gamma^2_t).
    \end{equation}
    Using \eqref{eq:4.12} and \eqref{eq:4.4}, the following convergence holds 
    \begin{equation}
        \label{eq:4.13}
          h^{t+1}_i|_{\cF_{t+1,t}} \xRightarrow{d} \sum_{j=0}^{t-1}\zeta_jh^{j+1}_i + \Gamma_tZ,
\end{equation} 
where $Z \sim \cN(0,1)$. Similar to Step 3d), we can write, using \eqref{eq:4.13}, $\bv^t_i \xRightarrow{d} \widetilde{\bv}^t_i$, where $\bv^t_i = (h^1_i,.., h^{t\+1}_i,x_{0i})$ and $\widetilde{\bv}^t_i = (h^1_i,\dots,h^{t}_i,\sum_{j=0}^{t-1}\zeta_jh^{j+1}_i + \Gamma_tZ,x_{0i})$. 
Hence, to prove \eqref{b:h} we first claim that $\Big|\limN \frac{1}{N}\sumN \phi_h(\widetilde{\bv}^t_i) - \E\Big[\phi_h(\tau_0Z_0,\dots,\tau_tZ_t,X_0)\Big] \Big| \das 0$.
Similar to {Step}~2d), using the triangular inequality, we verify that $\limN Y^t_1 \das 0$ and $\limN Y^t_2 \das 0$, where
$Y^t_1 = \Big|\frac{1}{N}\sumN \Big(\phi_h(\widetilde{\bv}^t_i) - \widetilde{\phi}_h(\bv^{t-1}_i)]\Big)  \Big|$ and $Y^t_2 = \Big| \frac{1}{N}\sumN \widetilde{\phi}_h(\bv^{t-1}_i) - \E\Big[\phi_h(\tau_0Z_0,\dots,\tau_tZ_t,X_0)\Big] \Big|$,
and $\widetilde{\phi}_h(\bv^{t-1}_i) = \E_Z[\phi_h(\widetilde{\bv}^t_i)]$, $\forall i$.

First, showing $\limN Y^t_1 \das 0$ is of interest.
By \eqref{eq:controlled2},     
$|\phi_h(\widetilde{\bv}^t_i)| \leq d^t_{1} \exp \Big( d^t_2 \Big( \sum_{j=0}^{t-1}|h^{j+1}_i|^{\lambda} + \Big| \sum_{j=0}^{t-1}\zeta_jh^{j+1}_i +\Gamma_t Z_i \Big|^{\lambda} +  |x_{0i}|^{\lambda} \Big) \Big)$, 
where $d^t_1 > 0$, $d^t_2>0$, and $1 \leq \lambda < 2$ are constants. 
{Using the inequality $\|\bx\|^{\lambda}_1 \leq (t+1)^{\lambda-1}\|\bx\|^{\lambda}_{\lambda}$ for $\bx \in \R^{(t+1) \times 1}$, we get $|\phi_h(\widetilde{\bv}^t_i)|$ $\leq d^t_{1} \exp \Big( d^t_2 \Big( \sum_{j=0}^{t-1}(1+(t+1)^{\lambda-1}|\zeta_j|^{\lambda})|h^{j+1}_i|^{\lambda} + (t+1)^{\lambda-1}|\Gamma_t|^{\lambda} |Z_i|^{\lambda} + |x_{0i}|^{\lambda} \Big) \Big)$. 
Hence, 
$\E_{Z}\Big[|\phi_b(\widetilde{\bv}^t_i)|^{2+\kappa}\Big] \leq d^t_5\exp\Big[d^t_4\Big(\sum_{j=0}^{t-1}|h^{j+1}_i|^{\lambda} + |x_{0i}|^{\lambda}\Big)\Big]$,
where $0 < \kappa <1 $, $d^t_4 = d^t_2(2+\kappa) \max \Big\{ 1+(t+1)^{\lambda-1}|\alpha_0|^{\lambda},\dots,1+(t+1)^{\lambda-1}|\alpha_{t-1}|^{\lambda}, (t+1)^{\lambda-1}|\Gamma_t|^{\lambda} \Big \}$, and $d^t_5 = (d^t_1)^{2+\kappa}\E_Z\Big[\exp(d^t_4|Z_i|^{\lambda})\Big]$ are constants}. 
Define $Y^t_{N,i} = \phi_h(\widetilde{\bv}^t_i) - \E_{Z}[\phi_h(\widetilde{\bv}^t_i)]$, $\forall i$. 
To prove the convergence $\limn Y^t_1 \das 0$, we will show that $\{Y^t_{N,i}\}_{i=1}^{N}$ satisfy the condition in {Lemma}~\ref{lm3} in Appendix~\ref{AppLM}. 
Indeed, the $\E_Z[|Y^t_{N,i}|^{2+\kappa}]$ is upper bounded as follows. 
    \small\begin{subequations}
    \beq 
    \E_Z[|Y^t_{N,i}|^{2+\kappa}]
    \d4&\leq&\d4  2^{1+\kappa}\bigg(\E_Z[|\phi_h(\widetilde{\bv}^t_i)|^{2+\kappa}] + |\E_Z[\phi_h(\widetilde{\bv}^t_i)]|^{2+\kappa} \bigg), \nonumber\\ \d4&\leq&\d42^{2+\kappa}\E_Z[|\phi_h(\widetilde{\bv}^t_i)|^{2+\kappa}] \nonumber \\
    \d4&\leq&\d4 d^t_6\exp\bigg(d^t_4\bigg(\sum_{j=0}^{t-1}|\zeta_j h^{j+1}_i|^{\lambda} + |x_{0i}|^{\lambda}\bigg)\bigg),\nonumber\\ 
    &\triangleq& {\psi_h(\bv^{t-1}_i).}\label{eq:AppB2}
    \eeq 
    \end{subequations}\normalsize
    Then, $\psi_h(\bv^{t-1}_i)$ is a controlled function. 
    From \eqref{eq:AppB2}, we get, for $N$ is sufficiently large,   
    \begin{subequations}
    \label{eq:AppB3}
    \beq 
    \frac{1}{N}\sumn \E_{Z}[|Y^t_{N,i}|^{2+\kappa}] &\leq& \frac{1}{N}\sumn \psi_h(\bv^{t-1}_i),\nonumber\\
    &\das& \E[\psi_h(\tau_0 Z_0,\dots,\tau_{t-1}Z_{t-1},X_0)],\nonumber\\
    &<& cN^{\kappa/2},\label{eq:AppB3c}
    \eeq 
    \end{subequations} 
    where $c$ is a positive constant and \eqref{eq:AppB3c} holds because $\E[\psi_b(\sigma_0\widetilde{Z}_0,\dots,\sigma_{t-1}\widetilde{Z}_{t-1},W)] = d^t_7 < \infty$ and there exists $N_t$, a positive constant, such that $d^t_7< cN^{\kappa/2}$ for $N > N_t$. 
    Using {Lemma}~\ref{lm3} in Appendix~\ref{AppLM}, $\limN \frac{1}{N}\sumN Y^t_{N,i} \das 0$, implying  $\limN Y^t_1 \das 0$. 

    We are now ready to verify the convergence $\limN Y^t_2 \das 0$. Applying the induction hypothesis \eqref{b:h} for $\widetilde{\phi}_b(\bv^{t-1}_i)$  gives 
 
\begin{multline}
\begin{aligned}
    \d4\d4\d4 &\limN \frac{1}{N}\sumN \widetilde{\phi}_h(\bv^{t-1}_i) 
    \das \E[\widetilde{\phi}_h(\tau_0Z_{0},\dots,\tau_{t-1}Z_{t-1},X_0)], \\ 
    &=\E\bigg[\E_Z[\phi_h(\tau_0Z_{0},\dots,\tau_{t-1}Z_{t-1},\sum_{j=0}^{t-1}\zeta_j\tau_jZ_j+\Gamma_tZ,X_0)] \bigg],\\
    &=\E\bigg[\phi_h(\tau_0Z_{0},\dots,\tau_{t-1}Z_{t-1},\sum_{j=0}^{t-1}\zeta_j\tau_jZ_j+\Gamma_tZ,X_0) \bigg].\nonumber
\end{aligned}
\end{multline}
Therefore, showing $\limN Y^t_2 \das 0$ is 
 equivalent to proving $\sum_{j=0}^{t-1}\zeta_j\tau_jZ_j+\Gamma_tZ = \tau_tZ_t$, where $Z_t \sim \cN(0,1)$ and $\tau_t$ is defined in \eqref{eq:SE}.
Similar to the proof of $\limn X^t_2 \das 0$ in Step 3d), setting $\phi_h(\bv^t_i) = (h^t_i)^2$, i.e., $\phi_h(\widetilde{\bv}^t_i) = \left( \sum_{j=0}^{t-1}\zeta_jh^{j+1}_i + \Gamma_tZ \right)^2$, we get  

\begin{multline}
    \begin{aligned}
        &\limN \lL \bh^{t+1},\bh^{t+1} \lR = \limN \frac{1}{N}\sumN \phi_h(\bv^t_i) \\ 
        &\ddis \limN \frac{1}{N}\sumN \phi_h(\widetilde{\bv}^t_i) \das  \E\bigg[\bigg(\sum_{j=0}^{t-1}\zeta_j\tau_jZ_j+\Gamma_tZ\bigg)^2\bigg].\nonumber
    \end{aligned}
\end{multline}
Using \eqref{eq:h2}, we get $\limN \lL \bh^{t+1},\bh^{t+1} \lR \das \limn \lL \bm^t,\bm^t \lR = \tau^2_t$, where the last equality holds by the induction hypothesis \eqref{b:b} for $\phi_b(\bu^t_i) = g^2_t(b^t_i,w_i)$, resulting in  $\limn \lL \bm^t,\bm^t \lR \das \E[g^2_t(\sigma_t\widetilde{Z}_t,W)] = \tau^2_t$. 
Hence, $\E\Big[\Big(\sum_{j=0}^{t-1}\zeta_j\tau_jZ_j+\Gamma_tZ\Big)^2\Big] \das \tau^2_t$, which implies $\sum_{j=0}^{t-1}\zeta_j\tau_jZ_j+\Gamma_tZ = \tau_tZ_t$. Thus, it is  verified that $\limN Y^t_2 \das 0$, which completes the proof of \eqref{b:h}.   
\end{document}